\documentclass[journal,twoside,web]{ieeecolor}
\usepackage{generic}
\usepackage{textcomp}
\usepackage{graphicx}
\usepackage{url}
\usepackage{caption}
\usepackage{subcaption}
\usepackage[english]{babel}
\usepackage{algorithm}
\usepackage[noend]{algpseudocode}

\usepackage{xspace}
\usepackage{amsmath}
\usepackage{amsfonts}
\usepackage{amssymb}
\usepackage[latin9]{inputenc}

\usepackage{amsthm}
\usepackage{siunitx,booktabs} %
\usepackage{multirow}
\usepackage{mathtools}          %
\usepackage{bm}
\usepackage{cite}
\mathtoolsset{showonlyrefs}

\newcommand{\reals}{\mathbb{R}}
\newcommand{\naturals}{\mathbb{N}}

\renewcommand{\Pr}{\textnormal{P}}

\newcommand{\f}{f}

\newcommand{\fith}{\f^{(i)}}

\newcommand{\x}{\textnormal{\textbf{x}}}   %
\newcommand{\y}{\textnormal{\textbf{y}}}
\renewcommand{\u}{\textnormal{\textbf{a}}}
\newcommand\indicator{\mathbf{1}}

\newcommand{\FullSet}{X}

\newcommand{\ROISet}{\mathcal{R}}
\newcommand{\ControlSet}{A}

\newcommand{\distBound}{\epsilon}
\newcommand{\noise}{\mathbf{v}}
\newcommand{\dataset}{\mathrm{D}}
\newcommand{\datasetsize}{d}
\newcommand{\distribution}{\mathcal{D}}
\newcommand{\noisedistribution}{p_{\mathbf{v}}}
\newcommand{\SubGaussianParam}{R}

\newcommand{\DiscretizationParam}{\Delta}

\newcommand{\str}{\pi}
\newcommand{\Str}{\Pi}
\newcommand{\adv}{\theta}
\newcommand{\Adv}{\Theta}

\newcommand{\strX}{\str_{\mathbf{y}}}

\newtheorem{theorem}{Theorem}
\newtheorem{problem}{Problem}

\newtheorem{assumption}{Assumption}
\newtheorem{proposition}{Proposition}

\newtheorem{lemma}{Lemma}
\newtheorem{definition}{Definition}
\newtheorem{remark}{Remark}

\newcommand{\compactSet}{W}

\newcommand{\kernel}{\kappa}
\newcommand{\mean}{\mu}
\newcommand{\meanpostith}{\mean^{(i)}_{\aimdp,\dataset}}
\newcommand{\varpostith}{\sigma^{(i)}_{\aimdp,\dataset}}
\newcommand\InfoGainBound{\gamma^{\datasetsize}_\kernel}
\newcommand\RKHS{\mathcal{H}}

\newcommand{\Safe}{\FullSet}
\newcommand{\Unsafe}{\bar{\Safe}}
\newcommand{\iaof}{\Longleftrightarrow}

\newcommand{\imdp}{\text{IMDP}\xspace}

\newcommand{\I}{\mathcal{I}}

\newcommand{\probDist}{\mathcal{D}}
\newcommand{\FeasibleDist}[2]{\mathfrak{T}_{#1}^{#2}}
\newcommand{\feasibleDist}[2]{\mathfrak{t}_{#1}^{#2}}

\newcommand{\Qimdp}{Q}
\newcommand{\Aimdp}{A}
\newcommand{\qimdp}{q}
\newcommand{\qimdppost}{Im(\qimdp, \aimdp)}

\newcommand{\qimdpprime}{\qimdp'}
\newcommand{\aimdp}{a}

\newcommand{\Spec}{\phi}

\newcommand{\APs}{AP}
\newcommand{\Prop}{\mathfrak{p}}
\newcommand{\StateFormula}{\phi}
\newcommand{\PathFormula}{\psi}
\newcommand{\Relation}{\bowtie}

\newcommand{\StateProb}{\Pr_{\Relation\p}}
\newcommand{\StateProbComp}{\Pr_{\bar{\Relation} 1-\p}}
\newcommand\NextOp{\mathcal{X}}
\newcommand\UntilOp{\mathcal{U}}
\newcommand\BoundedUntilOp{\UntilOp^{\leq k}}
\newcommand\EventuallyOp{\mathcal{F}}
\newcommand\BoundedEventuallyOp{\EventuallyOp^{\leq k}}
\newcommand\GloballyOp{\mathcal{G}}
\newcommand\Satisfies{\models}
\newcommand\NotSatisfies{\nvDash}

\newcommand{\Pup}{\hat{\textnormal{P}}}
\newcommand{\Plow}{\check{\textnormal{P}}}
\newcommand{\Qyes}{\Qimdp^{\text{yes}}}
\newcommand{\Qno}{\Qimdp^{\text{no}}}
\newcommand{\Qposs}{\Qimdp^{\text{?}}}

\newcommand{\pathmdp}{\omega}
\newcommand{\pathimdp}{\omega}
\newcommand{\pathmdpfin}{\pathmdp^{\mathrm{fin}}}
\newcommand{\pathimdpfin}{\pathmdpfin}

\newcommand{\Pathimdp}{\mathit{Paths}}
\newcommand{\Pathimdpfin}{\mathit{Paths}^{\mathrm{fin}}}
\newcommand{\last}{\textit{last}}

\newcommand{\p}{p}

\newcommand{\plow}{\check{\p}}
\newcommand{\pup}{\hat{\p}}

\DeclareMathOperator{\sinc}{sinc}

\def\BibTeX{{\rm B\kern-.05em{\sc i\kern-.025em b}\kern-.08em
    T\kern-.1667em\lower.7ex\hbox{E}\kern-.125emX}}
\begin{document}

\title{Formal Verification of Unknown Dynamical Systems via Gaussian Process Regression}

\author{John Skovbekk, \IEEEmembership{Student Member, IEEE}, Luca Laurenti, Eric Frew, and Morteza Lahijanian, \IEEEmembership{Member, IEEE}
\thanks{Submitted for review on 06/26/2023. This work was supported in part by the NSF grant 2039062 and NSF Center for Unmanned Aircraft Systems under award IIP-1650468. (\emph{Corresponding author: John Skovbekk})}
\thanks{John Skovbekk is with the Smead Aerospace Engineering Sciences Dept. at the University of Colorado, 3775 Discovery Dr, Boulder, CO 80305 USA (e-mail: john.skovbekk@colorado.edu)}
\thanks{Luca Laurenti is with the Delft Center for Systems and Control at TU Delft, Netherlands (e-mail:l.laurenti@tudelft.nl)}
\thanks{Eric Frew is with the Smead Aerospace Engineering Sciences Dept. at the University of Colorado, Boulder, CO 80305 USA (e-mail: eric.frew@colorado.edu).}
\thanks{Morteza Lahijanian is with the Smead Aerospace Engineering Sciences Dept. at the University of Colorado, Boulder, CO 80305 USA (e-mail: morteza.lahijanian@colorado.edu).}
}

\maketitle

\begin{abstract}
Leveraging autonomous systems in safety-critical scenarios requires verifying their behaviors in the presence of uncertainties and black-box components that influence the system dynamics.
In this work, we develop a framework for verifying discrete-time dynamical systems with unmodelled dynamics and noisy measurements against temporal logic specifications from an input-output dataset.
The verification framework employs Gaussian process (GP) regression to learn the unknown dynamics from the dataset and abstracts the continuous-space system as a finite-state, uncertain Markov decision process (MDP).
This abstraction relies on space discretization and transition probability intervals that capture the uncertainty due to the error in GP regression by using reproducible kernel Hilbert space analysis as well as the uncertainty induced by discretization.
The framework utilizes existing model checking tools for verification of the uncertain MDP abstraction against a given temporal logic specification. 
We establish the correctness of extending the verification results on the abstraction created from noisy measurements to the underlying system.
We show that the computational complexity of the framework is polynomial in the size of the dataset and discrete abstraction.  
The complexity analysis illustrates a trade-off between the quality of the verification results and the computational burden to handle larger datasets and finer abstractions.
Finally, we demonstrate the efficacy of our learning and verification framework on several case studies with linear, nonlinear, and switched dynamical systems. 
\end{abstract}


\section{Introduction}
\label{sec:intro}
Recent advances in technology have led to a rapid growth of autonomous systems operating in \emph{safety-critical} domains. Examples include self-driving vehicles, unmanned aircraft, and surgical robotics.
As these systems are given such delicate roles, it is \emph{essential} to provide guarantees on their performance.
To address this need, formal verification offers a powerful framework with rigorous analysis techniques \cite{Clarke99,BaierBook2008}, which are traditionally model-based.
An accurate dynamics model for an autonomous system, however, may be unavailable due to, e.g., black-box components, or so complex that existing verification tools cannot handle.
To deal with such shortcomings, machine learning offers capable methods that can identify models solely from data.
While eliminating the need for an accurate model, these learning methods often lack quantified guarantees with respect to the latent system~\cite{hullermeier2021aleatoric}.
The gap between model-based and data-driven approaches is in fact the key challenge in verifiable autonomy.
This work focuses on closing this gap by developing a data-driven verification method that can provide formal guarantees for systems with unmodelled dynamics.

Formal verification of continuous control systems has been widely studied, e.g., \cite{tabuada2009verification,Belta:Book:2017,doyen2018verification,KD01,soudjani2015fau,Lahijanian:TAC:2015,laurenti2020formal}.
These methods are typically based on model checking algorithms \cite{Clarke99,BaierBook2008}, which check whether a finite-state model satisfies a given specification.
The specification language is usually a form of temporal logic, which provides rich expressivity.
Specifically, probabilistic \textit{linear temporal logic} (LTL) and \textit{probabilistic computation tree logic} (PCTL) are used to define specifications for stochastic systems \cite{Clarke99,BaierBook2008}.
To bridge the gap between continuous and discrete domains, a finite-state \textit{abstraction} is constructed.
This takes the form of a finite Markov process if the latent system is stochastic \cite{Lahijanian:CDC:2012,cauchi2019efficiency,laurenti2020formal}, and
the resulting frameworks admit strong formal guarantees.
Nevertheless, these frameworks are model-based and cannot be employed for analysis of systems with unknown models.

Machine learning has emerged as a powerful tool for learning unknown functions from data.
In particular, Gaussian process (GP) regression is becoming widespread to learn dynamical systems due to its predictive power, uncertainty quantification, and ease of use \cite{Rasmussen:Book:2003,srinivas2012information}.
By conditioning a prior GP on a dataset, GP regression returns a posterior distribution that predicts the system dynamics via the application of the Bayes rule.
However, the main challenge in using GP regression (and machine learning in general) for verification in the context of safety-critical applications is the need to formally quantify the error in the learning process and propagate it in the verification pipeline. 
Existing data-driven approaches often lack the required formalism and/or are limited to linear systems. %

\begin{figure}[tp]
    \centering
    \includegraphics[width=0.48\textwidth,trim={1.15cm 0.5cm 1.5cm 0.5cm},clip]{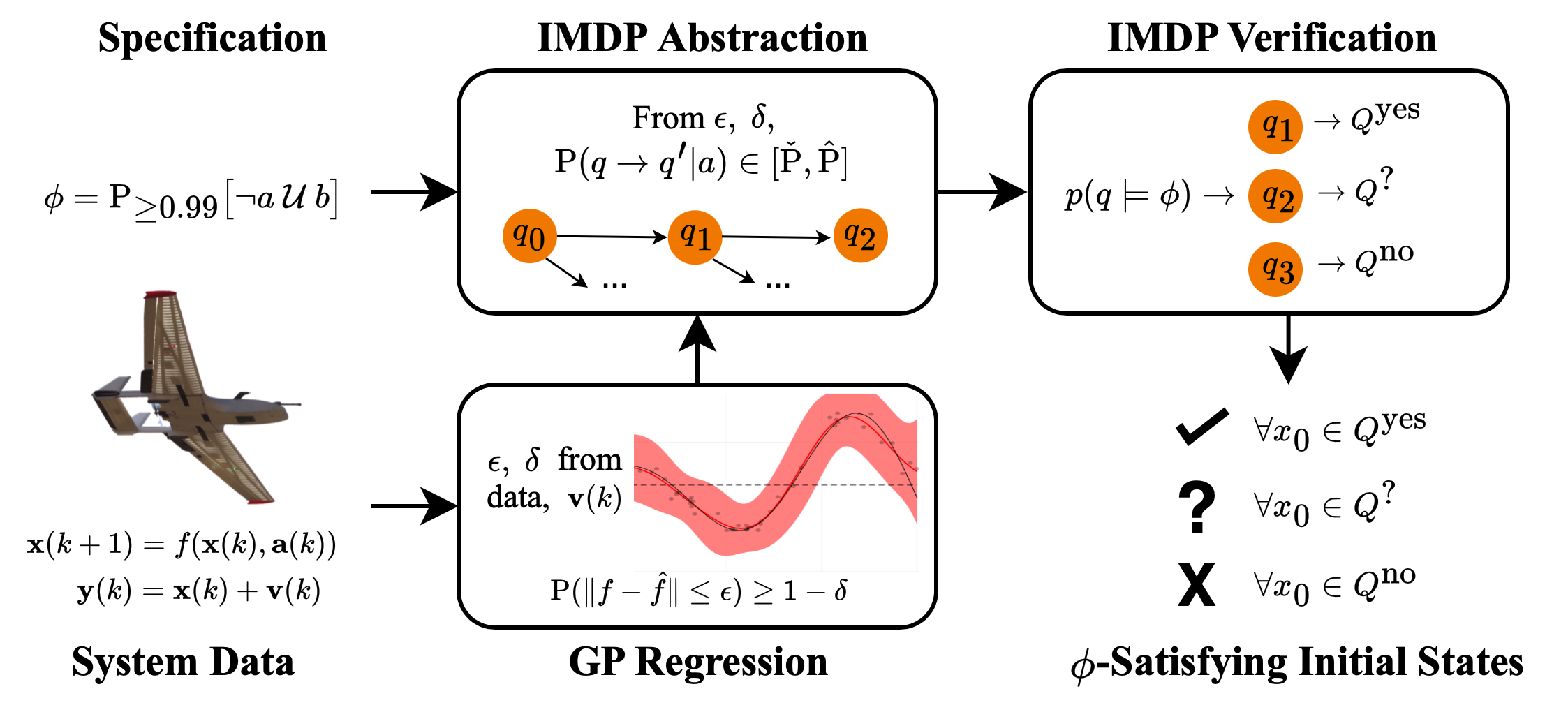}
    \caption{Overview of our learning-based, formal verification framework.
    } 
    \label{fig:overview}
\end{figure}

This article introduces an end-to-end data-driven verification framework that enables the use of off-the-shelf model-based verification tools and provides the required formalism to extend the guarantees to the latent system.
An overview of the framework is shown in Figure~\ref{fig:overview}.
Given a noisy input-output dataset of the latent system and a temporal logic specification (LTL or PCTL formula), the framework produces verification results in the form of three sets of initial states: $\Qyes$, $\Qposs$, $\Qno$, from which the system is guaranteed to fully, possibly, and never satisfy the specification, respectively.

Beginning with the dataset, the framework learns the dynamics via GP regression and formally quantifies the distance between the learned model and the latent system.  
For this step, we build on the GP regression error analysis in \cite{chowdhury2017kernelized} and derive formal bounds for the constants in the error term.
Specifically, by relying on the relation between reproducing kernel Hilbert spaces (RKHS) and kernel functions, we establish upper bounds for the RKHS constant of a function and the information gain. %
Using these bounds, the framework then constructs a finite-state, uncertain Markov decision process (MDP) abstraction that includes probabilistic bounds for every possible behavior of the system.
Next, it uses an existing model checking tool to verify the abstraction against a given PCTL or LTL formula.  
Here, we specifically focus on PCTL formulas since a model checking tool for uncertain MDPs is readily available \cite{Lahijanian:TAC:2015}. 
Finally, the framework closes the verification loop by mapping the resulting guarantees to the latent system.

The main contribution of this work is a \emph{data-driven formal verification framework for unknown dynamical systems}.
The novelties include: 
\begin{enumerate}
    \item formalization of the GP regression error bound in \cite{chowdhury2017kernelized} by deriving bounds for its challenging constants, 
    \item an abstraction procedure including derivation of optimal transition probability bounds,
    \item a proof of correctness of the final verification results for the latent system, and
    \item an illustration of the framework on several case studies with linear, nonlinear, and switched dynamical systems as well as empirical analysis of parameter trade-offs.
\end{enumerate}

\subsection{Related Work}

Recently, there has been a surge in applying formal verification methods to dynamic systems with machine learning components~\cite{perrouin2018learning,xiang2018verification}.
We begin with a brief overview of model checking for verification, and then discuss its application to learning-based approaches for dynamic systems.

Verification of continuous-space systems via model checking  requires constructing a finite abstraction with a simulation~\cite{Girard:ITAC:2007,Lahijanian:CDC:2012}.
An appropriate and widely-used abstraction model for continuous stochastic systems is the \textit{Interval-valued MDP} (IMDP)~\cite{Lahijanian:CDC:2012,Lahijanian:TAC:2015,Luna:WAFR:2014,luna:icra:2014,Luna:AAAI:2014}.
The advantage of IMDPs is their capability to incorporate multiple sources of uncertainty as well as the availability of their model checking tools with PCTL and LTL specifications~\cite{Hahn:QEST:2017,Hahn:TOMACS:2019,Lahijanian:TAC:2015}.
For a system with unknown dynamics and non-Gaussian measurement noise, however, it is not clear how to construct such an abstraction.
In this work, we address this problem by studying a method for formalization of all the sources of uncertainties to enable IMDP abstraction construction from data with formal guarantees.

Data-driven verification methods can be categorized to linear and nonlinear assumptions on the latent system dynamics.
Existing work on linear systems include methods to verify parameterized, time-invariant dynamics subject to performance specifications~\cite{Haesaert:Automatica:2017,Kenanian:Automatica:2019,Salamati:Automatica:2021,badings2022sampling}, to perform system identification with statistical guarantees~\cite{dean2020sample}, and to monitor the safety of stochastic linear systems with unmodelled errors~\cite{yoon2019predictive}.
For nonlinear systems, 
existing methods are largely based on sampling the latent system and using statistical bounding to generate safety or reachability guarantees~\cite{Fan2017,Wang2021,herbert2021scalable,musavi2021hoover,sun2022neureach,Nejati2023FormalVerification,Alanwar2023Reachability}.
Specifically, the DryVR framework relies on learning discrepancy functions for an unknown, deterministic system~\cite{Fan2017}, and NeuReach learns a neural-network for computing probabilistic reachable sets by sampling deterministic system trajectories~\cite{sun2022neureach}.
Recently, jointly performing reinforcement learning and formal verification can guarantee safety and stability for parameterized polynomial systems from sampled trajectories~\cite{wang2023joint}.
Those results, however, are limited to reachability property analysis from the sampled initial conditions, and many consider noiseless measurements of the system.

For complex objectives beyond safety and reachability, existing approaches construct finite-state abstractions from observed system data~\cite{jackson2020safety,jin2022data,jiang2022safe,hashimoto2022learning,peruffo2022data}.
These methods require rigorous error quantification of the data-driven models to guarantee the correctness of the verification results.
Such errors can be quantified with deterministic bounds~\cite{hashimoto2022learning}, \emph{probably-approximately correct} (PAC) bounds~\cite{xue2020pac,musavi2021hoover,jin2022data,sun2022neureach,badings2022sampling,Nejati2023FormalVerification}, or Bayesian probabilistic bounds~\cite{jackson2020safety,jiang2022safe}. 
Among these techniques, PAC-based methods provide straightforward results, but they rely heavily on the number of samples and primarily do not consider measurement noise to the best of our knowledge.
In addition, unlike Bayesian methods, the guarantees they provide are conditioned on a confidence that is not rooted in a probability measure on the system.  
In this work, we consider unknown nonlinear systems with measurement noise, and our approach is Bayesian-based and provides guarantees based on well-defined probability measure over system trajectories.

To quantify probabilistic uncertainty, GP regression is widely employed due to its universal approximation property and mature theory~\cite{Rasmussen:Book:2003}.
For unknown dynamic systems, GP regression has been used to learn maximal invariant sets in reinforcement learning~\cite{akametalu2014reachability,sui2015safe,berkenkamp2016safe,berkenkamp2017safe}, for runtime control and safety monitoring of dynamic systems~\cite{Helwa2019,yoon2019predictive}, and to construct barrier certificates for verifying system safety~\cite{jagtap2020barriers,wajid2022formal}.
A limitation of GP regression method is that 
the correctness of the predicted distribution is conditioned on that fact that the latent function is drawn from the prior process, which may be too restrictive as an assumption for various applications \cite{srinivas2012information}. Furthermore, it only supports Gaussian additive noise~\cite{Rasmussen:Book:2003}.
These limitations can be relaxed 
if the latent function lies in the \emph{Reproducing Kernel Hilbert Space} (RKHS) defined by the GP kernel. Then, the noise can be 
non-Gaussian, and more importantly, the GP error can be quantified probabilistically~\cite{chowdhury2017kernelized,berkenkamp2017safe,jagtap2020barriers,hashimoto2022learning,wajid2022formal}
for arbitrary latent continuous functions. 
A closely-related work uses deterministic RKHS-based GP errors to construct a deterministic transition system and verify against LTL specifications when there is no measurement noise~\cite{hashimoto2022learning}.
Nevertheless, that and many other existing works that use the RKHS setting rely on ad-hoc parameter approximations of the error term, which, strictly speaking, revokes the correctness of the resulting guarantees.  
In this work, we also employ RKHS-based approach but provide a formal method for determining the error parameters, enabling hard guarantees.

\section{Problem Formulation}
\label{sec:problem}
\noindent
We consider a discrete-time controlled process given by 
\begin{align}
    \label{eq:system}
    \begin{split}
    \x(k+1) &= \f(\x(k), \u(k)) \\
    \y(k) &= \x(k)+ \noise(k),
    \end{split}
\end{align}
\begin{equation*}
    k\in \naturals,\;\;\; \x(k) \in \reals^n, \;\;\; \u(k) \in \ControlSet, \;\;\; \y(k) \in \reals^n, \;\;\; 
\end{equation*}
where $\ControlSet=\{a_1,\dots,a_{|\ControlSet|}\}$ is a finite set of actions, $\f:\reals^n\times \, \ControlSet \rightarrow\reals^n$ is a (possibly non-linear) function that is \emph{unknown}, $\y$ is a measurement of $\x$ with noise $\noise(k)$, a random variable with probability density function $\noisedistribution$.
We assume that $\noisedistribution$ is a stationary and \emph{conditionally $\SubGaussianParam$-sub-Gaussian} distribution\footnote{Conditional $\SubGaussianParam$-Sub-Gaussian distributions are those whose tail decay at least as fast as the tail of  normal distribution with variance $\SubGaussianParam^2$ conditioned on the filtration up to the previous step.
This class of distributions includes the Gaussian distribution itself and distributions with bounded support~\cite{massart2007concentration}.}
and the $\SubGaussianParam$ parameter of the distribution is known, but the knowledge of the explicit form of the distribution is not necessary.
Intuitively,  $\y$ is a stochastic process whose behavior depends on the latent process $\x$, which itself is driven by actions from $\ControlSet$.

The evolution of Process~\eqref{eq:system} is described using trajectories of the states and measurements. A finite  \emph{state trajectory} up to step $k$ is a sequence of state-action pairs denoted by $\omega_{\x}^k= x_0 \xrightarrow{a_0} x_1 \xrightarrow{a_1}  \ldots \xrightarrow{a_{k-1}} x_k$.  
The state trajectory $\omega_{\x}^k$ induces a measurement (or observation) trajectory $\omega_{\y}^k= y_0 \xrightarrow{a_0} y_1 \xrightarrow{a_1}  \ldots \xrightarrow{a_{k-1}} y_k$ via the measurement noise process.  
We denote infinite-length state and measurement trajectories by $\omega_{\x}$ and $\omega_\y$ and the set of all state and measurement trajectories by $\Omega_{x}$ and $\Omega_{y}$, respectively.
Further, we use $\omega_{\x}(i)$ and $\omega_{\y}(i)$ to denote the $i$-th elements of a state and measurement trajectory, respectively.

\subsection{Probability Measure}

We assume that the action at time $k$ is chosen by (unknown) 
\textit{control strategy} $\strX:\Omega_y^k\to\ControlSet$ based on the measurement trajectory up to time $k$. 
Note that, even though $\strX$ is deterministic, given a finite state trajectory, the next action chosen by $\strX$ is stochastic due to the noise process $\noise,$ which induces a probability distribution over the trajectories of the latent process $\x$. 
In particular, given a fixed initial condition $x_0 \in \mathbb{R}^n$ and a strategy $\strX$, we denote with  $\omega_{y,v_0,\ldots ,v_{k}}^{k,\strX,x_{{0}}}$  the measurement trajectory with noise at time $0\leq i \leq k$ fixed to $v_i$ and relative to the inducing state trajectory $\Omega^k_x$ such that $\Omega^k_x(0)=x_{0}$.\footnote{Note that $\Omega^k_x$ is uniquely defined by its initial state $x_0$ and the value of the noise at the various time steps, i.e,   $v_0,...,v_{k}.$ } 
 We can then define the probability space  $(\Omega_x^k,\mathcal{B}(\Omega_x^k),\Pr),$ where $\mathcal{B}(\Omega_x^k)$ is the $\sigma$-algebra on $\Omega_x^k$ generated by the product topology, and $\Pr$ is a probability measure on the sets in $\mathcal{B}(\Omega_x^k)$ such that for a set $X \subset
\mathbb{R}^n$ \cite{bertsekas2004stochastic}:
\begin{align*}
    &\Pr(\omega_{\mathbf{x}}(0)\in X)=\indicator_{X}(x_0),\\
    &\Pr(\omega_{\mathbf{x}}(k)\in X \mid \omega_{\mathbf{x}}(0)=x_{0}) = \\
    & \hspace{22mm} \int...\int T\big(X \mid x,\strX(\omega_{y,v_0,\ldots ,v_{k-1}}^{k-1,\strX,x_{0}})\big) \\
    &\hspace{35mm}\noisedistribution(v_{k-1}) \ldots\noisedistribution(v_0)d v_{k-1}\ldots d v_0,
\end{align*}
where 
\begin{equation*}
    \indicator_{X}(x) = 
    \begin{cases}
        1 & \text{if } x\in X\\
        0 &\text{otherwise}
        \end{cases}
\end{equation*}
is the indicator function, and
\begin{equation}
    \label{eq:transition-kernel}
    T(X \mid x,a)=\Pr(\omega_{\x}(k)\in X \mid \omega_{\x}(k-1)=x,\u_{k-1}=a)
\end{equation}
is the transition kernel that defines the single step probabilities for Process \eqref{eq:system}. In what follows, we call $t(\bar{x} \mid x,a)$ the density function associated to the transition kernel, i.e.,  
\begin{equation*}
    \label{eq:transition-kernel-density}
    T(X \mid x,a)=\int_X t(\bar{x} \mid x,a)d\bar{x}.
\end{equation*}
Intuitively, $\Pr$ is defined by marginalizing $T$ over all possible observation trajectories. Note that as the initial condition $x_0$ is fixed, the marginalization is over the distribution of the noise up to time $k-1.$
 We remark that $\Pr$ is also well defined for $k=\infty$ by the \emph{Ionescu-Tulcea extension theorem}~\cite{abate2014effect}. 

The above definition of $\Pr$ implies that
if $f$ were known, the transition probabilities of Process \eqref{eq:system} are Markov and deterministic  once the action is known. 
However, we stress that, even in the case where $\aimdp$ is known, the exact computation of the above probabilities is infeasible in general as $\f$ is unknown. 
Next, we introduce continuity assumptions that allow us to estimate $\f$ 
and consequently compute bounds on the transition kernel $T$.

\subsection{Continuity Assumption}
Taking function $\f$ as completely unknown leads to an ill-posed problem.
Hence, we assume $f$ belongs to a \emph{reproducing kernel Hilbert space} (RKHS), which is a Hilbert space of functions that lie in the span of a positive-definite kernel function.
This is a standard assumption\cite{srinivas2012information},~\cite{chowdhury2017kernelized} that constrains $\f$ to be a well-behaved analytical function on a compact (closed and bounded) set.

\begin{assumption}[RKHS Continuity]
    \label{assump:rkhs}
    For a compact set $\compactSet\subset \mathbb{R}^n$, let $\kernel:\mathbb{R}^n\times \mathbb{R}^n\to \mathbb{R}$ be a given kernel and $\mathcal{H}_\kernel(\compactSet)$ the reproducing kernel Hilbert space (RKHS) of functions over $\compactSet$ corresponding to $\kernel$~\cite{srinivas2012information}. 
    Let the RKHS-norm of $\f$ be $\| \f\|_\kernel$ (defined in Lemma~\ref{lemma:rkhsnorm} in Appendix).
    Then, for each $a \in \ControlSet$  and $i\in \{1,\ldots, n \},$ $\f^{(i)}(\cdot,a) \in \mathcal{H}_\kernel(\compactSet)$ and for a constant $B_i>0,$ $\| \f^{(i)}(\cdot,a) \|_\kernel \leq B_i$, where $\f^{(i)}$ is the $i$-th component of $\f$.
\end{assumption}
\noindent
Although Assumption~\ref{assump:rkhs} limits $\f$ to a class of analytical functions, it is not overly restrictive as we can choose $\kernel$ such that $\mathcal{H}_\kernel$ is dense in the space of continuous functions~\cite{steinwart2001influence}.
For instance, this holds for the widely-used squared exponential kernel function~\cite{Rasmussen:Book:2003}.
This assumption allows us to use GP regression to estimate $\f$ and leverage error results as we detail in Section~\ref{sec:gp-reg}.

\subsection{PCTL Specifications}
\label{sec:PCTL}
We are interested in the behavior of the latent process $\x$ as defined in \eqref{eq:system} over regions of interest $\ROISet =\{r_1,\dots,r_{|\ROISet|}\}$ with $r_i\subset \reals^n$. 
For example, these regions may indicate target sets that should be visited or unsafe sets that must be avoided both in the sense of physical obstacles and state constraints (e.g., velocity limits). In order to define properties over $\ROISet$, for a given state $x$ and region $r_i$, we define an atomic proposition $\Prop_i$ to be true ($\top$) if $x\in r_i$, and otherwise false ($\bot$).  
The set of atomic propositions is given by $\APs=\{\Prop_1, \dots, \Prop_{|\ROISet|}\}$, and the label function $L:\reals^n\to 2^{\APs}$ returns the set of atomic propositions that are true at each state.

Probabilistic Computation Tree Logic (PCTL) \cite{BaierBook2008} is a formal language that allows for the expression of complex behaviors of stochastic systems.

We begin by defining the syntax and semantics of PCTL specifications.
\begin{definition}[PCTL Syntax]
    \label{def:pctl-syntax}
    Formulas in PCTL are recursively defined over the set of atomic propositions $\APs$ in the following manner:
    \begin{align*}
       \text{State Formula} \quad \StateFormula &\coloneqq \top \;|\; \Prop \;|\; \neg\StateFormula \;|\; \StateFormula\wedge\StateFormula \;|\; \StateProb[\PathFormula]\\
        \text{Path Formula} \quad \PathFormula &\coloneqq \NextOp\StateFormula \;|\; \StateFormula\, \BoundedUntilOp\StateFormula \;|\;  \StateFormula\, \UntilOp\StateFormula 
    \end{align*}
    where $\Prop \in \APs$, 
    $\neg$ is the negation operator, $\wedge$ is the conjunction operator, $\StateProb$ is the probabilistic operator, 
    $\Relation \in \{\leq, <, \geq, >\}$ is a relation placeholder, and $p\in [0,1]$. 
    The temporal operators are the $\NextOp$ (Next), $\BoundedUntilOp$ (Bounded-Until) with respect to time step $k\in\naturals$, and $\UntilOp$ (Until).
\end{definition}
\begin{definition}[PCTL Semantics]
\label{def:pctl-semantics}
The satisfaction relation $\Satisfies$ is defined inductively as follows. 
For state formulas,
\begin{itemize}
    \item $x\Satisfies \top$ for all $x\in \reals^n$; 
    \item $x\Satisfies\Prop\iaof\Prop\in L(x)$;
    \item $x\Satisfies(\StateFormula_1\wedge\StateFormula_2)\iaof (x\Satisfies\StateFormula_1)\wedge(x\Satisfies\StateFormula_2)$;
    \item $x\Satisfies\neg\StateFormula\iaof x\NotSatisfies\StateFormula$;
    \item $x\Satisfies\StateProb[\PathFormula]\iaof p^x(\PathFormula)\Relation p$, where $p^x(\PathFormula)$ is the probability that all infinite trajectories initialized at $x$ satisfy $\PathFormula$.
\end{itemize}
For a state trajectory $\omega_{\x}$, the satisfaction relation $\models$ for path formulas is defined as:
\begin{itemize}
    \item $\omega_{\x}\Satisfies\NextOp\StateFormula\iaof \omega_{\x}(1)\Satisfies\StateFormula$;
    \item $\omega_{\x}\Satisfies\StateFormula_1\BoundedUntilOp\StateFormula_2\iaof\exists i\leq k$ s.t. $\omega_{\x}(i)\Satisfies\StateFormula_2\wedge \omega_{\x}(j)\Satisfies\StateFormula_1 \; \forall j\in[0,i)$;
    \item $\omega_{\x}\Satisfies\StateFormula_1\UntilOp\StateFormula_2\iaof\exists i\geq 0$ s.t. $\omega_{\x}(i)\Satisfies\StateFormula_2\wedge \omega_{\x}(j)\Satisfies\StateFormula_1 \; \forall j\in[0,i)$.
\end{itemize}
\end{definition}

The common operators bounded eventually $\BoundedEventuallyOp$ and eventually $\EventuallyOp$ are defined respectively as $\StateProb[\BoundedEventuallyOp\StateFormula]\equiv\StateProb[\top\BoundedUntilOp\StateFormula]$, and $\StateProb[\EventuallyOp\StateFormula]\equiv\StateProb[\top\UntilOp\StateFormula].$ 
The globally operators $\GloballyOp^{\leq k}$ and $\GloballyOp$ are defined as 
$\StateProb[\GloballyOp^{\leq k}\StateFormula]\equiv \StateProbComp[\EventuallyOp^{\leq k} \neg \StateFormula]$, and $\StateProb[\GloballyOp\StateFormula]\equiv \StateProbComp[\EventuallyOp \neg \StateFormula],$ where 
$\bar \bowtie$ indicates the opposite relation, i.e, $\bar < \equiv \; >$, $\bar \leq \equiv \; \geq$, $\bar \geq \equiv \; \leq$, and $\bar > \equiv \; <$. 

As an example the property, \emph{``the probability of reaching the target by while avoiding unsafe areas, and with probability at least 0.95 of eventually reaching the secondary target if the breach area is entered, is at least 0.95''} can be expressed with the PCTL formula 
\begin{align*}
    \StateFormula = \Pr_{\geq 0.95}\big[ \big( \neg\Prop_\text{unsafe} \wedge \big(&\Prop_\text{breach}\implies\\
    &\Pr_{\geq 0.95}[\EventuallyOp  \Prop_\text{secondary}] \big)\big) \, \UntilOp \, \Prop_\text{target} \big].
\end{align*}

\subsection{Problem Statement}
The verification problem asks whether Process~\eqref{eq:system} satisfies a PCTL formula $\Spec$ under all control strategies.
It can also be posed as finding the set of initial states in a compact set $\FullSet\subset\reals^n$, from which Process~\eqref{eq:system} satisfies $\Spec$.

In lieu of an analytical form of 
$\f$, we assume to have a dataset
$\dataset=\{(x_i,a_i,y_i)\}_{i=1}^{\datasetsize}$
generated by Process~\eqref{eq:system} where $y_i$ is the noisy measurement of $\f(x_i, a_i)$, and $\datasetsize \in \naturals^+$.
Using this dataset, we can infer $\f$ and reason about the trajectories of Process~\eqref{eq:system} 
(via Assumption~\ref{assump:rkhs}). 
The formal statement of the problem considered in this work is as follows.

\begin{problem}\label{prob:main}
Let $\FullSet\subset\reals^n$ be a compact set,  
$\ROISet$ a set of regions of interest, and $\APs$ its corresponding set of atomic propositions (as defined in Section~\ref{sec:PCTL}).
Given a dataset $\dataset$ generated by Process~\eqref{eq:system} and a PCTL formula $\Spec$ defined over $\APs$, determine the initial states $\FullSet_0 \subset \FullSet$ from which Process \eqref{eq:system}
satisfies $\Spec$ without leaving $\FullSet$.
\end{problem}

Our approach to Problem~\ref{prob:main} is based on constructing an IMDP abstraction (formally defined in Sec.~\ref{sec:prelimIMDP}) of the latent process $\x$ using GP regression and checking if the abstraction satisfies $\Spec$. 
We begin by estimating the unknown dynamics with GP regression and, by virtue of Assumption~\ref{assump:rkhs}, we leverage existing probabilistic error bounds on this estimate. We present upper bounds on the constants required to compute this probabilistic error and derive bounds on the transition kernel in Eq.~\eqref{eq:transition-kernel} over a discretization of the space $\FullSet$.
Then, we construct the abstraction of Process~\eqref{eq:system} in the form of an uncertain (interval-valued) Markov decision process, which accounts for the uncertainties in the regression of $\f$ and the discretization of $\FullSet$. 
Finally, we perform model checking on the abstraction  to get sound bounds on the probability that $\mathbf{x}$ satisfies $\Spec$ from any initial state.

\begin{remark} 
    We remark that our IMDP abstraction construction is general, and hence, can be used for the verification of Process~\eqref{eq:system} against properties in other specification languages such as LTL. We specifically focus on PCTL  properties mainly because of the availability of its model checking tool for IMDPs \cite{Lahijanian:TAC:2015}.
\end{remark}

\section{Preliminaries}
\label{sec:preliminaries}

\subsection{Gaussian Process Regression}\label{sec:gp-reg}

Gaussian process (GP) regression~\cite{Rasmussen:Book:2003} aims to estimate an unknown function $\mathrm{f}:\reals^n\to \mathbb{R}$ from a dataset $\dataset=\{(\mathrm{x}_i, \mathrm{y}_i)\}_{i=1}^{\datasetsize}$ where $\mathrm{y}_i = \mathrm{f}(\mathrm{x}_i) + \mathrm{v}_i$ and $\mathrm{v}_i$ is a sample of a normal distribution with zero mean and $\sigma_{\noise}^2$ variance, denoted by $\mathrm{v}_i\sim\mathcal{N}(0, \sigma_{\noise}^2)$.
The standard assumption is that $\mathrm{f}$ is a sample from a prior GP with zero-valued mean function and kernel function $\kernel_0$.\footnote{Extensions with non-zero mean are a trivial generalization \cite{Rasmussen:Book:2003}.}

Let $X$  and $Y$ be ordered vectors  with all points in $\dataset$ such that $X_i = \mathrm{x}_i$ and $Y_i = \mathrm{y}_i$.  
Further, let $K$ denote the matrix with element $K_{i,j}=\kernel_0(\mathrm{x}_i,\mathrm{x}_j)$, $k(\mathrm{x},X)$ the vector such that $k_{i}(\mathrm{x},X)=\kernel_0(\mathrm{x},X_i)$, and $k(X,\mathrm{x})$ defined accordingly. 
Predictions at a new input point $\mathrm{x}$ are given by the conditional distribution of the prior at $\mathrm{x}$  given $\dataset$, which is still Gaussian and with mean $\mean_\dataset$ and variance $\sigma_\dataset^2$ given by
\begin{align}
\begin{split}\label{eq:post-mean}
      &\mean_{\dataset}(\mathrm{x}) = k(\mathrm{x},X)^T \big( K+ \sigma_{\mathrm{v}}^2 I \big)^{-1} Y 
\end{split}\\
\begin{split}\label{eq:post-kernel}
      & \sigma_{\dataset}^2(\mathrm{x}) = \kernel_0(\mathrm{x},\mathrm{x})-
      k(\mathrm{x},X)^T\big( K+ \sigma_{\mathrm{v}}^2 I \big)^{-1}k(X,\mathrm{x})
\end{split}
\end{align}
where $I$ is the identity matrix of size $\datasetsize$. 

As mentioned previously, the noise in Process~\eqref{eq:system} is sub-Gaussian, so the standard assumption that $\f$ is a sample from the prior GP cannot be used.  
Hence, we cannot directly use \eqref{eq:post-mean} and \eqref{eq:post-kernel} to bound the latent function. 
Rather, as common in the literature~\cite{srinivas2012information}, we rely on the relationship between GP regression and the reproducing kernel Hilbert space (RKHS).

\subsection{Error Quantification}\label{sec:rkhs}

The reliance on a positive-definite kernel function is the basis for relating a GP with kernel $\kernel$ with an RKHS $\RKHS_\kernel$.
For universal kernels (such as the squared exponential function), the associated RKHS is dense in the continuous functions on any compact set \cite{micchelli2006unikernels}.

Given Assumption~\ref{assump:rkhs}, the following proposition bounds the GP learning error when $f$ is a function in $\RKHS_{\kernel}$, without posing any distributional assumption on $f$.

\begin{proposition}[\hspace{1sp}\cite{chowdhury2017kernelized}, Theorem 2]
    \label{th:RKHS}
    Let $\compactSet$ be a compact set, $\delta\in(0,1]$, and $\dataset$ be a given dataset generated by $\mathrm{f}$ with cardinality $|\dataset|=d$. Further, let $\Gamma>0$ be a bound on the maximum information gain of $\kernel$, i.e., $\InfoGainBound\leq\Gamma$, and $B>0$ such that $\|\mathrm{f} \|_{\kernel}\leq B$. Assume that $\mathrm{v}$ is $\SubGaussianParam$-sub-Gaussian and that $\mean_\dataset$ and $\sigma_\dataset$ are found by setting $\sigma_{\mathrm{v}}^2=1+2/\datasetsize$ and using \eqref{eq:post-mean} and \eqref{eq:post-kernel}. 
    Define $\beta(\delta)=
    B + \SubGaussianParam\sqrt{2(\Gamma + 1 + \log{1/\delta})}$. Then, it holds that 
    \begin{align}
        \label{eq:proberror}
        \Pr\big(\forall \mathrm{x} \in \compactSet, \, |\mean_\dataset(\mathrm{x}) - \mathrm{f}(\mathrm{x})| \leq \beta(\delta)\sigma_\dataset(\mathrm{x}) \big)\geq 1-\delta.
    \end{align}
\end{proposition}

Proposition~\ref{th:RKHS} assumes that the noise on the measurements is $\SubGaussianParam$-sub-Gaussian, which is more general than Gaussian noise. 
Nevertheless, the probability bound~\eqref{eq:proberror} depends on bounding the constants $\|\mathrm{f}\|_{\kernel}$ and $\InfoGainBound$ which are described as \emph{challenging to compute} in the literature~\cite{lederer2019uniform} and are often bounded according to heuristics without guarantees~\cite{berkenkamp2016safe,jagtap2020barriers}.  
In this work, we derive formal upper bounds for each of these terms in Section~\ref{subsec:rkhs-bounds}.

\subsection{Interval Markov Decision Processes}
\label{sec:prelimIMDP}
An interval Markov decision process (\imdp) is a generalization of a Markov decision process, where the transitions under each state-action pair are defined by probability intervals~\cite{givan2000bounded}.
\begin{definition}[Interval MDP] \label{def:imdp}
    An interval Markov decision process is a tuple $\I = (\Qimdp,\Aimdp,\Plow,\Pup,\APs,L)$ where
    \begin{itemize}
    	\setlength\itemsep{1mm}
    	\item $\Qimdp$ is a finite set of states,
    	\item $\Aimdp$ is a finite set of actions, where $\Aimdp(\qimdp)$ is the set of available actions at state $\qimdp\in \Qimdp$,
        \item $\Plow: \Qimdp \times \Aimdp \times \Qimdp \to [0,1]$ is a function, where $\Plow(\qimdp,\aimdp,\qimdpprime)$ defines the lower bound of the transition probability from state $\qimdp\in \Qimdp$ to state $\qimdpprime\in \Qimdp$ under action $\aimdp \in \Aimdp(\qimdp)$,
        \item $\Pup: \Qimdp \times \Aimdp \times \Qimdp \to [0,1]$ is a function, where $\Pup(\qimdp,\aimdp,\qimdpprime)$ defines the upper bound of the transition probability from state $\qimdp$ to state $\qimdpprime$ under action $\aimdp \in \Aimdp(\qimdp)$,
        \item $\APs$ is a finite set of atomic propositions,
        \item $L: \Qimdp \to 2^{\APs}$ is a labeling function that assigns to each state $\qimdp$ possibly several elements of $\APs$.
    \end{itemize}
\end{definition}
\noindent
For all $\qimdp,\qimdpprime \in \Qimdp$ and $\aimdp \in \Aimdp(\qimdp)$, it holds that $\Plow(\qimdp,\aimdp,\qimdpprime) \leq \Pup(\qimdp,\aimdp,\qimdpprime)$ and
\begin{equation*}
    \sum_{\qimdpprime \in \Qimdp} \Plow(\qimdp,\aimdp,\qimdpprime) \leq 1 \leq \sum_{\qimdpprime \in \Qimdp} \Pup(\qimdp,\aimdp,\qimdpprime).
\end{equation*}
Let $\distribution(\Qimdp)$ denote the set of probability distributions over $\Qimdp$.  
Given $\qimdp \in \Qimdp$ and $\aimdp \in \Aimdp(\qimdp)$, we call $\feasibleDist{\qimdp}{\aimdp} \in \probDist(\Qimdp)$ a \textit{feasible distribution} over states reachable from $\qimdp$ under $\aimdp$ if the transition probabilities respect the intervals defined for each possible successor state $\qimdpprime$, i.e., $\Plow(\qimdp,\aimdp,\qimdpprime)\leq\feasibleDist{\qimdp}{\aimdp}(\qimdpprime)\leq\Pup(\qimdp,\aimdp,\qimdpprime)$. 
We denote the set of all feasible distributions for state $\qimdp$ and action $\aimdp$ by $\FeasibleDist{\qimdp}{\aimdp}$. 

A path $\pathimdp$ of an \imdp is a sequence of state-action pairs $\pathimdp = \qimdp_0 \xrightarrow{\aimdp_0} \qimdp_1 \xrightarrow{\aimdp_1} \qimdp_2 \xrightarrow{\aimdp_2}  \ldots$ 
such that $\aimdp_i \in \Aimdp(\qimdp_i)$ and $\Pup(\qimdp_i, \aimdp_i, \allowbreak{\qimdp_{i+1}}) > 0$ (i.e., transitioning is possible)
for all $i \in \naturals$.
We denote the last state of a finite path $\pathimdpfin$ by
$\last(\pathimdpfin)$ and the set of all finite and infinite paths by 
$\Pathimdpfin$ and $\Pathimdp$, respectively. 
Actions taken by the \imdp are determined by a choice of strategy $\str$ which is defined below.

\begin{definition}[Strategy]
\label{def:strategy}
    A strategy $\str$ of an \imdp $\I$ is
    a function $\str: 
    \Pathimdpfin 
    \to \Aimdp$ 
    that maps a finite path
    $\pathimdpfin$ of $\I$ onto an 
    action in $\Aimdp$. The set of all strategies is denoted by $\Str$.
\end{definition}

Once an action is chosen according to a strategy, a feasible distribution needs to be chosen from $\FeasibleDist{\qimdp}{\aimdp}$ to enable a transition to the next state.  
This task falls on the adversary function $\adv$ as
defined below. 
\begin{definition}[Adversary]
\label{def:adversary}
    Given an \imdp $\I$, an adversary is a function $\adv: \Pathimdpfin \times \Aimdp \rightarrow \probDist(\Qimdp)$ that, for each  finite path $\pathimdpfin \in \Pathimdpfin$ and action $\aimdp \in \Aimdp(\last(\pathimdpfin))$, chooses a feasible distribution $\feasibleDist{\qimdp}{\aimdp} \in \FeasibleDist{\last(\pathimdpfin)}{\aimdp}$. 
    The set of all adversaries is denoted by $\Adv$.
\end{definition}

Once a strategy is selected, the \imdp becomes an interval Markov chain. Further, choosing an adversary results in a standard Markov chain.
Hence, given a strategy and an adversary, a probability measure can be defined on the paths of the $\imdp$ via the probability of paths on the resulting Markov chain~\cite{Lahijanian:TAC:2015}.

\section{IMDP Abstraction}
\label{sec:abstraction}
\noindent
To solve Problem~\ref{prob:main}, we begin by constructing a finite abstraction of Process~\eqref{eq:system} in the form of an \imdp that captures the state evolution of the system under known actions.
This involves partitioning the set $\FullSet$ into discrete regions and determining the transition probability intervals between each pair of discrete regions under each action to account for the uncertainty due to the learning and discretization processes.

\subsection{\imdp States and Actions}\label{sec:imdp-construction}
In the first step of the abstraction, the set $\FullSet$ is discretized into a finite set of non-overlapping regions $\Qimdp_{\Safe}$ such that 
$\bigcup_{ \qimdp\in \Qimdp_{\Safe}}~q = \FullSet.$
The discretization must maintain consistent labelling with the regions of interest $\ROISet = \{r_i\}_{i=1}^{|\ROISet|}$ where $r_i\subseteq \FullSet$, i.e., (with an abuse of the notation of $L$), 
$$L(\qimdp)=L(r_i) \quad  \text{ iff } \quad L(x)=L(r_i) \;\; \forall x\in\qimdp.$$
To ensure this consistency, the regions of interest are used as the foundation of the discretization.
Then, we define 
$$\Qimdp=\Qimdp_{\Safe} \cup \{\qimdp_{\Unsafe}\}, \quad \text{ where } \quad \qimdp_{\Unsafe} = \reals^n\setminus \FullSet.$$ 
The regions in $\Qimdp$ are associated with the states of the \imdp, and with an abuse of notation, $\qimdp \in \Qimdp$ indicates both a state of the \imdp and the associated discrete region. 
Finally, the \imdp action space is set to the action set $\Aimdp$ of Process~\eqref{eq:system}. 
 
\subsection{Transition Probability Bounds}\label{subsec:tbounds}

The \imdp transition probability intervals are pivotal to abstracting Process~\eqref{eq:system} correctly.
For all $\qimdp,\qimdp' \in \Qimdp$ and action $\aimdp \in \Aimdp$, the intervals should bound the true transition kernel
\begin{align*}
    \Plow(\qimdp, \aimdp, \qimdp') &\leq \min_{x\in \qimdp} T(\qimdp'\mid x, \aimdp), \\ 
    \Pup(\qimdp, \aimdp, \qimdp') &\geq \max_{x\in \qimdp} T(\qimdp'\mid x, \aimdp). 
\end{align*} 
The bounds on the transition kernel $T$ must account for the uncertainties due to unknown dynamics and space discretization. We use GP regression and account for regression errors to derive these bounds for every $(\qimdp, \aimdp, \qimdp')$ tuple.

\subsubsection{IMDP State Images and Regression Error}\label{sec:overapp}

We perform GP regression using the given dataset $\dataset$ and analyze the evolution of each \imdp state under the learned dynamics.
In addition, the associated learning error is quantified using Proposition~\ref{th:RKHS}. 

Each output component of the dynamics under an action is estimated by a separate GP, making a total of $n \, |\Aimdp|$ regressions.
Let $\meanpostith(x)$ and $\varpostith(x)$ respectively denote the posterior mean and covariance functions for the $i$-th output component under action $\aimdp$ obtained via GP regression.
These GPs are used to evolve $\qimdp$ under action $\aimdp$, defined by image
$$Im(\qimdp, \aimdp)=\big\{ \meanpostith(x) \mid x\in\qimdp, \;i\in[1,n] \big\},$$
which describes the evolution of all $x\in\qimdp$.

The \emph{worst-case} regression error is
\begin{equation}
    e^{(i)}(\qimdp, \aimdp)= \sup_{x\in\qimdp}|\meanpostith(x) - \fith(x,\aimdp)|,
\end{equation}
which provides an error upper-bound for all continuous states $x \in \qimdp$.
To use Proposition~\ref{th:RKHS} for probabilistic reasoning on this worst-case error, the supremum of the posterior covariance in $\qimdp$ denoted by 
$$\overline{\sigma}^{(i)}_{\aimdp,\dataset}(q)= \sup_{x\in \qimdp}\varpostith(x)$$
is used.  
Both the posterior image and covariance supremum can be calculated for each \imdp state using GP interval bounding as in~\cite{blaas2020adversarial}.
Then, Proposition~\ref{th:RKHS} is applied with a scalar $\distBound^{(i)}\geq 0$,
\begin{equation}\label{eq:prob}
    \Pr\big(e^{(i)}(\qimdp, \aimdp)\leq \distBound^{(i)}\big)\geq 1-\delta
\end{equation}
where $\delta$ satisfies $\distBound^{(i)}=\beta(\delta) \overline{\sigma}^{(i)}_{\aimdp,\dataset}(q)$.

\subsubsection{Bounds on the RKHS Parameters}\label{subsec:rkhs-bounds}
The computation of $\Pr\big( e^{(i)}(\qimdp,\aimdp) \leq \distBound^{(i)}\big)$ via Proposition \ref{th:RKHS} requires two constants: the RKHS norm $\| \fith \|_\kernel$ and information gain $\InfoGainBound$.
In this section, we provide formal upper bounds for these constants.
\begin{proposition}[RKHS Norm Bound]\label{prop:rkhsbound}
Let $\fith\in\RKHS$ where $\RKHS$ is the RKHS defined by the kernel function $\kernel$. 
Then
\begin{equation}
    \label{eq:RKHS_norm}
    \|\fith\|_{\kernel}\leq \frac{\sup_{x\in  \FullSet}| \fith(x)|}{\inf_{x,x'\in  \FullSet}\kernel(x,x')^{\frac{1}{2}}}
\end{equation}
for all $x,x'\in  \FullSet$. 
\end{proposition}

\noindent
The proof is provided in the Appendix, which relies on the relationship between the RKHS norm and the kernel function.

Proposition~\ref{prop:rkhsbound} bounds the RKHS norm by the quotient of the supremum of the function and the infimum of the kernel function on a compact set. 
Note that \eqref{eq:RKHS_norm} requires the square root of the kernel in the denominator. Hence, so long the chosen kernel is a positive function (e.g., squared exponential), we can compute a finite upper bound for the RKHS norm.
As $\FullSet$ is a finite-dimensional compact set, it can be straightforward to calculate bounds as in the following remark. 

\begin{remark} 
    An upper bound for the numerator of \eqref{eq:RKHS_norm} can be estimated using the Lipschitz constant $L_{\fith}$ of $\fith$ (or an upper bound of it). 
    That is, the supremum in the numerator can be bounded by
    \begin{equation}
        \sup_{x\in \FullSet} |\fith(x)| \leq |\fith(x')| + L_{\fith}\text{diam}( \FullSet).
    \end{equation}
    where $x'$ is any point in $\FullSet$, $\text{diam}(\FullSet)\coloneqq \sup_{x,x'\in\FullSet}\|x-x'\|$ is the diameter of $\FullSet$.
    The term $|\fith(x')|$ can be bounded by sampling at a single point $x'$ if the measurement noise has bounded support or can be set using a known equilibrium point in $\FullSet$, i.e., a point that satisfies $f(x') = x'$ and using the components therein.
    The Lipschitz constant itself can be estimated from data with statistical guarantees~\cite{knuth2021planning,chakrabarty2020safe}.
\end{remark}

The information gain $\InfoGainBound$ can be bounded using the of size of the dataset and the GP hyperparameters as in the following proposition.

\begin{proposition}[Information Gain Bound]\label{prop:infobound}
 Let $\datasetsize$ be the size of the dataset $\dataset$ and $\sigma_\mathrm{v}$ the parameter used for GP regression. Then, the maximum information gain is bounded by
    \begin{equation}
        \label{eq:info_gain}
        \InfoGainBound \leq \datasetsize\log(1+\sigma_\mathrm{v}^{-2}s),    
    \end{equation}
where $s=\sup_{x\in  \FullSet} \kernel(x,x)$. 
\end{proposition} 
\noindent The proof is provided in the Appendix and relies on the parameters for regression and positive-definiteness of $\kernel$ 
to employ Hadamard's inequality. %
While this bound on $\InfoGainBound$ is practical, it can be further improved by using more detailed knowledge of the kernel function, such as its spectral properties~\cite{vakili2021information}.

\subsubsection{Transitions between states in $\Qimdp_{\FullSet}$}

To reason about transitions using the images of regions, the following notations are used to indicate the expansion and reduction of a region and the intersection of two regions. 

\begin{definition}[Region Expansion and Reduction]
    Given a compact set (region) $\qimdp\subset\reals^{n}$ and a set of $n$ scalars $c = \{c_1,\dots,c_n\}$, where $c_i \geq 0$, the expansion of $\qimdp$ by $c$ is defined as
    \begin{align*}
        \overline{\qimdp}(c) = \{x\in\reals^n \mid \exists x_\qimdp \in \qimdp \;\; s.t. \;\; |x_q^{(i)}-&x^{(i)}|\leq c_i
        \;\;\\&\forall i=\{1,\dots,n\} \},
    \end{align*}
    and the reduction of $\qimdp$ by $c$ is 
    \begin{align*}
        \underline{\qimdp}(c) = \{x_\qimdp \in \qimdp \mid  \forall x_{\partial \qimdp} \in \partial \qimdp,  \;\; |x_\qimdp^{(i)}-x_{\partial \qimdp}^{(i)}|&> c_i 
        \;\;\\ 
        &\forall i=\{1,\dots,n\}\},
    \end{align*}
    where $\partial \qimdp$ is the boundary of $\qimdp$.
\end{definition}

The intersection between the expanded and reduced states indicate the possibility of transitioning between the states. 
This intersection function is defined as
\begin{equation*}
    \indicator_V(W) = \begin{cases} 1 &\text{ if } V\cap W \neq \emptyset\\ 0 & \text{ otherwise }\end{cases}
\end{equation*}
for sets $V$ and $W$.

The following theorem presents bounds for the transition probabilities of Process \eqref{eq:system} that account for both the regression error and induced discretization uncertainties using Proposition~\ref{th:RKHS} and sound approximations using the preceding definitions. 

\begin{theorem}
\label{Th:TransitionBounds}
    Let $\qimdp, \qimdp' \in \Qimdp_{\FullSet}$ and $\distBound$, $\distBound'\in\reals^n$ be non-negative vectors.
    For any action $\aimdp \in \Aimdp$, the transition kernel is bounded by
	\begin{align*}
		& \min_{x\in q} T(q' \mid x, a)
		\\ &\quad\quad \geq\Big(1- \indicator_{\FullSet\setminus \underline{\qimdp}'(\distBound) }\big(\qimdppost\big)\Big)
		\prod_{i=1}^n \Pr\Big( e^{(i)}(\qimdp,\aimdp) \leq \distBound^{(i)}\Big), \quad\;\;
	\end{align*}
	and
	\begin{align*}
	    & \max_{x\in q}T( q' \mid x,a) 
	    \\ &\quad\quad \leq 1-\prod_{i=1}^n\Pr\big( e^{(i)}(\qimdp,\aimdp) \leq \distBound'^{(i)}\big)\Big(1-\indicator_{ \overline{\qimdp}'(\distBound') }\big(\qimdppost\big)\Big).
	\end{align*}
\end{theorem}

\begin{proof}
The proof of Theorem~\ref{Th:TransitionBounds} relies on bounding the probability of transitioning from $\qimdp$ to $\qimdpprime$ conditioned on the learning error $\distBound$ ($\distBound'$) given by Eq.~\eqref{eq:prob}, and using the expanded (reduced) image of $\qimdp$ with $\distBound$ ($\distBound'$) to find a point in $\qimdp$ that minimizes (maximizes) this bound.
We provide the proof for the upper bound on $T$, and the lower bound follows analogously. 

Let $\qimdp'$ denote subsequent states of $\qimdp$, and suppose $\omega_{\x}(k)\in \qimdp$ for an arbitrary timestep $k$.
The probability of transitioning to $\qimdp'$ under action $\aimdp$ is upper-bounded beginning with transition kernel and using the law of total probability conditioned on the learning error:
\begin{align*}
 & \max_{x\in \qimdp}\Pr( \omega_{\x}(k+1) \in \qimdp' \mid \omega_{\x}(k)=x,\aimdp)\\
     &=\max_{x\in \qimdp}\big( \Pr( \omega_{\x}(k+1)  \in \qimdp' \wedge e(\qimdp,\aimdp)\leq \epsilon \mid \omega_{\x}(k)=x ) \; +\\
     &\qquad \Pr( \omega_{\x}(k+1) \in \qimdp' \wedge e(\qimdp,\aimdp)> \epsilon \mid \omega_{\x}(k)=x)  \big)\\
     &=\max_{x\in \qimdp}\big( \Pr(\omega_{\x}(k+1) \in \qimdp'  \mid e(\qimdp,\aimdp)\leq \epsilon, \omega_{\x}(k)=x )\cdot \\
     &\quad \Pr(e(\qimdp,\aimdp) \leq \epsilon )+ \\
     & \quad\Pr(\omega_{\x}(k+1) \in \qimdp'  \mid e(\qimdp,\aimdp) > \epsilon, \omega_{\x}(k)=x )\Pr(e(\qimdp,\aimdp) > \epsilon ) \big)
\end{align*}
Next, $\Pr(\omega_{\x}(k+1) \in \qimdp'  \mid e(\qimdp,\aimdp)\leq \epsilon, \omega_{\x}(k)=x )$ is upper bounded using the intersection indicator, $\Pr(\omega_{\x}(k+1) \in \qimdp'  \mid e(\qimdp,\aimdp) > \epsilon, \omega_{\x}(k)=x )$ is upper bounded by one, and $\Pr(e(\qimdp,\aimdp) > \epsilon )$ is replaced by the equivalent $1-\Pr(e(\qimdp,\aimdp) \leq \epsilon )$:
\begin{align*}
     &\leq \max_{x\in \qimdp} \Big( \mathbf{1}_{ \overline{\qimdp}'} \big(Im(\qimdp,\aimdp)\big) \, \Pr(e(\qimdp,\aimdp)\leq \epsilon ) + \\
     &\hspace{40mm}1 \cdot \big( 1-  \Pr(e(\qimdp,\aimdp)\leq \epsilon )  \big) \Big)\\
     &=\max_{x\in \qimdp}\Big( \mathbf{1}_{ \overline{\qimdp}'} \big(Im(\qimdp,\aimdp)\big) \prod_{i=1}^n\Pr(e^{(i)}(\qimdp,\aimdp)\leq \epsilon_i ) \, + \\
         &\hspace{40mm} 1-  \prod_{i=1}^n\Pr( e^{(i)}(\qimdp,\aimdp) \leq  \epsilon_i )  \Big),
\end{align*}
where the last inequality is due to fact that the components of $\noise$ are mutually independent.
Rearranging these terms, we obtain the upper bound on $T$.
\end{proof}

The width of the transition probability intervals in Theorem~\ref{Th:TransitionBounds} relies on the choices of $\distBound$ and $\distBound'$.
The following proposition provides the optimal values for both $\distBound$ and $\distBound'$.

\begin{proposition}
\label{prop:distbound}
    Let $\partial W$ denote the boundary of a compact set $W\subset\reals^n$.
     The distance between the upper and lower bounds in Theorem~\ref{Th:TransitionBounds} is minimized if $\distBound'=\distBound$ and $\distBound$ is chosen such that for each $i\in[1,n]$,
     \begin{align*}
        \distBound^{(i)} &= \inf_{x\in \partial \qimdppost, \, x'\in \partial \qimdp'} |x^{(i)}-x'^{(i)}|. 
    \end{align*}
\end{proposition}

\noindent
The proof for Proposition~\ref{prop:distbound} is provided in the Appendix. %

The intuition for this proposition is as follows.
Consider the image $\qimdppost$ and the intersection $\qimdppost\cap \qimdpprime$. There are three possible outcomes; (1) the intersection is empty, (2) the intersection is non-empty but not equal to $\qimdppost$, and (3) the intersection is equal to $\qimdppost$.
By examining Theorem~\ref{Th:TransitionBounds}, it is clear that if (2) is true, then we get trivial transition probability bounds of $[0,1]$. In this case, the choice of $\distBound$ does not matter. If (1) is true, then the bounds are $[0, \Pup(\qimdp,\aimdp,\qimdp')]$
and we have an opportunity to get a non-trivial upper bound. The best upper bound is achieved by choosing large $\distBound$ to capture more regression error, but small enough to keep the indicator function zero. This is achieved by choosing $\distBound$ according to Proposition~\ref{prop:distbound}, which corresponds to minimizing the L1 norm of $x-x'$. Finally, if (3) is true, then the bounds are $[\Plow(\qimdp,\aimdp,\qimdp'), 1]$ and the approach is similar.

\subsubsection{Transitions to $\qimdp_{\Unsafe}$}
Transition intervals to the unsafe state $\qimdp_{\Unsafe}$ are the complement of the transitions to the full set $\FullSet$, or  
\begin{align*}
\Plow(\qimdp,\aimdp,\qimdp_{\Unsafe}) &= 1 - \Pup(\qimdp,\aimdp,\FullSet)
\\\Pup(\qimdp,\aimdp,\qimdp_{\Unsafe}) &= 1 - \Plow(\qimdp,\aimdp,\FullSet)
\end{align*}
where $\Plow(\qimdp,\aimdp,\FullSet)$ and $\Pup(\qimdp,\aimdp,\FullSet)$ are calculated with Theorem~\ref{Th:TransitionBounds}. 
Finally, the unsafe state has an enforced absorbing property where
$
\Plow(\qimdp_{\Unsafe},\aimdp,\qimdp_{\Unsafe})=\Pup(\qimdp_{\Unsafe},\aimdp,\qimdp_{\Unsafe})=1.
$

The \imdp abstraction of Process~\eqref{eq:system} incorporates the uncertainty due to the measurement noise, the uncertainty of $\f$, and the induced error from discretizing the continuous state space.
The width of the transition intervals quantifies the conservativeness of the \imdp abstraction with respect to the underlying system. 
In the worst-case, a transition interval of $[0,1]$ indicates that the transition is possible but provides no meaningful information about the true transition probability. Likewise, smaller intervals indicate the abstraction better models the underlying system.

\section{Verification}
\label{sec:verification}

In this section, \imdp verification against PCTL specifications is summarized, and the correctness of the verification results on the \imdp abstraction is established via Theorem \ref{th:correctness}.

\subsection{\imdp Verification}
\label{subsec:modelchecking}

PCTL model checking of an \imdp is a well established procedure~\cite{Lahijanian:TAC:2015}.  
For completeness, we present a summary of it here. 
Specifically, we focus on the probabilistic operator $\Pr_{\bowtie p}[\PathFormula]$, where $\PathFormula$ includes the bounded until ($\BoundedUntilOp$) or unbounded until ($\UntilOp$) operator, since the procedure for the next ($\NextOp$) operator is analogous~\cite{Lahijanian:TAC:2015}.  
In what follows, $\PathFormula$ is assumed to be a path formula with $\BoundedUntilOp$, which becomes $\UntilOp$ when $k = \infty$.

For an initial state $\qimdp\in\Qimdp$ and PCTL path formula $\PathFormula$, let $\Qimdp^0$ and $\Qimdp^1$ be the set of states from which the probability of satisfying $\PathFormula$ is 0 and 1, respectively.  
These sets are determined simply by the labels of the states in $\Qimdp$. 
Further, let $\plow^k(\qimdp)$ and $\pup^k(\qimdp)$ be respectively the lower-bound and upper-bound probabilities that the paths initialized at $\qimdp$ satisfy $\PathFormula$ in $k$ steps. 
These bounds are defined recursively by
\begin{align}
    \plow^k(\qimdp) &= 
    \begin{cases}
        1 & \text{if } \qimdp \in \Qimdp^1\\
        0 & \text{if } \qimdp \in \Qimdp^0\\
        0 & \text{if } \qimdp \notin (\Qimdp^0 \cup \Qimdp^1) \wedge k = 0\\
        \mathrlap{\min\limits_{\aimdp} \min\limits_{\adv_{\qimdp}^{\aimdp}} \sum_{\qimdp'} \adv_{\qimdp}^{\aimdp}(\qimdp') \plow^{k-1}(\qimdp')} & \hspace{40mm} \text{otherwise}
    \end{cases} 
    \label{eq:plow} 
    \\
    \pup^k(\qimdp) &= 
    \begin{cases}
        1 & \text{if } \qimdp \in \Qimdp^1\\
        0 & \text{if } \qimdp \in \Qimdp^0\\
        0 & \text{if } \qimdp \notin (\Qimdp^0 \cup \Qimdp^1) \wedge k = 0\\
        \mathrlap{\max\limits_{\aimdp} \max\limits_{\adv_{\qimdp}^{\aimdp}} \sum_{\qimdp'} \adv_{\qimdp}^{\aimdp}(\qimdp') \pup^{k-1}(\qimdp')} & \hspace{40mm}\text{otherwise}
    \end{cases} 
    \label{eq:pup}
\end{align}
In short, the procedure exactly finds the minimizing and maximizing adversaries with respect to the equations above for each state-action pair, and then determines the minimizing and maximizing action for each state.
It is guaranteed that each probability bound converges to a fix value in a finite number of steps. 
Hence, this procedure can be used for both bounded until ($\BoundedUntilOp$) and unbounded until ($\UntilOp$) path formulas.
See~\cite{Lahijanian:TAC:2015} for more details.
The final result is the probability interval 
$$[\plow^k(\qimdp),\pup^k(\qimdp)] \subseteq [0,1]$$ 
of satisfying $\PathFormula$ within $k$ time steps for each \imdp state $\qimdp \in \Qimdp$.
Similar to the transition intervals, the widths of the satisfaction probability intervals are a measure of the conservativeness of the verification results.

Given the PCTL formula $\Spec=\Pr_{\bowtie p}[\PathFormula]$, the satisfaction intervals are used to classify states as belonging to $\Qyes$, $\Qno$, or $\Qposs$, i.e., those states that satisfy, violate and possibly satisfy $\Spec$. 
For example, when the relation $\bowtie$ in formula $\Spec$ is $>$ and $\p$ is the threshold value inherent to $\Spec$, 
\begin{equation*}
    \qimdp\in \begin{cases}
      \Qyes&\text{if }\plow^k(\qimdp)> \p\\
      \Qno&\text{if }\pup^k(\qimdp)\leq \p\\
      \Qposs&\text{otherwise.}
    \end{cases}
\end{equation*}
Thus, $\Qyes$ consist of the states that satisfy $\Spec$, and $\Qno$ consists of the states that do no satisfy $\Spec$ for every choice of adversary and action.
$\Qposs$ is the set of indeterminant states for which no guarantees can be made with respect to $\Spec$ due to large uncertainty.

\subsection{Verification Extension and Correctness}
The final task is to extend the \imdp verification results to Process~\eqref{eq:system} even though the \imdp does not model the measurement noise on the system.
This is possible by observing that the value iteration procedures in \eqref{eq:plow} and \eqref{eq:pup} are solved by finding the \emph{extreme} actions at each state. All other actions have outcomes that lie in the satisfaction probability intervals.
The following theorem asserts that these intervals, defined for $\qimdp$, bound the probability that all paths initialized at $x\in\qimdp$ satisfy $\PathFormula$.

\begin{theorem}
    \label{th:correctness}
    Let $\qimdp\in \Qimdp$ be both a region in $\FullSet$ and a state of the \imdp abstraction constructed on dataset $\dataset$, and
    $\plow(\qimdp)$ and $\pup(\qimdp)$ be the lower- and upper-bound probabilities of satisfying $\PathFormula$ from $\qimdp$ computed by the procedure in Section~\ref{subsec:modelchecking}.
    Further, let $x \in \qimdp$ be a point and
    $\Pr(\omega_{\x}\Satisfies\PathFormula\mid \omega_{\x}(0) = x,\strX)$ be the probability that all paths initialized at $x$ satisfy $\PathFormula$ under strategy $\strX$ given $\dataset$.
    Then, it holds that
    \begin{equation*}
        \begin{split}
            \Pr(\omega_{\x}\Satisfies\PathFormula\mid \omega_{\x}(0) = x,\strX) &\geq \plow(\qimdp), \\ 
            \Pr(\omega_{\x}\Satisfies\PathFormula\mid \omega_{\x}(0) = x,\strX) &\leq  \pup(\qimdp)
        \end{split}
    \end{equation*}
    for every strategy $\strX$.
\end{theorem}
\begin{proof}
    We present the proof of the lower bound;  the proof for the upper bound is analogous. 
    Furthermore, in what follows for the sake of a simpler notation, we assume that $\strX$ is stationary, i.e., $\strX:\reals^{n}\to A.$ The case of time-dependent strategies follows similarly.  
    The following lemma shows a value iteration that computes the probability of satisfying $\PathFormula$ from initial state $x$ under $\strX$ in $k$ steps.
    \begin{lemma}
        \label{lemma:ValueIterOriginalSys}
        Let $\Qimdp^0$ and $\Qimdp^1$ be sets of discrete regions that satisfy $\PathFormula$ with probability 0 and 1, respectively, and $\FullSet^0 = \cup_{\qimdp \in \Qimdp^0} \, \qimdp$ and $\FullSet^1 = \cup_{\qimdp \in \Qimdp^1} \, \qimdp$.  Further,
        let $V_k^{\strX}:\reals^n \to \reals $ be defined recursively as
        \begin{align*}
            V_k^{\strX}(x)=
            \begin{cases}
                1 & \text{if } x \in \FullSet^1\\
                0 & \text{if } x \in \FullSet^0\\
                0 & \text{if } x \notin (\FullSet^0 \cup \FullSet^1) \wedge k = 0\\
                \mathrlap{\mathbf{E}_{v\sim \noisedistribution, \, \bar{x}\sim T(\cdot \mid x,\strX(x+v) )} [ V_{k-1}(\bar{x})]} & \hspace{45mm}\text{otherwise.}
            \end{cases}
        \end{align*}    
        Then, it holds that $ V_k^{\strX}(x) = \Pr(\omega^k_{\x}\Satisfies\PathFormula\mid \omega_{\x}(0) = x,\strX).$
    \end{lemma}
    The proof of Lemma~\ref{lemma:ValueIterOriginalSys} is in the Appendix.
Hence, to conclude the proof of Theorem~\ref{th:correctness}, it suffices to show that for $x \in \qimdp$, $V_k^{\strX}(x)\geq \plow^k(\qimdp)$ for every $k\geq 0$ and every $\strX$. 
This can be shown by induction similarly to Theorem 4.1 in \cite{delimpaltadakis2021abstracting}. 
The base case is $k=0$, where it immediately becomes clear that 
$$ V_0^{\strX}(x)=\mathbf{1}_{\FullSet^1}(x)=\mathbf{1}_{\Qimdp^1}(\qimdp)= \plow^0(\qimdp).$$
For the induction step we need to show that under the assumption that  for every $q\in \Qimdp$
$$ \min_{x\in \qimdp} V_{k-1}^{\strX}(x)\geq \plow^{k-1}(\qimdp)$$
holds, 
then $V_k^{\strX}(x)\geq \plow^k(\qimdp).$ Assume for simplicity $x \notin (\FullSet^0 \cup \FullSet^1)$, then we have
\begin{align*}
    V&_k^{\strX}(x) = \mathbf{E}_{v\sim \noisedistribution, \bar{x}\sim T(\cdot \mid x,\strX(x+v) )} [ V_{k-1}(\bar{x})] \\ 
    & = \int \int V_{k-1}^{\strX}(\bar{x})t(\bar{x}\mid x,\strX(x+v))\noisedistribution(v)d\bar{x}dv\\
    & =  \int \bigg(\sum_{\qimdp\in\Qimdp}\int_{\qimdp} V_{k-1}^{\strX}(\bar{x})t(\bar{x}\mid x,\strX(x+v))d\bar{x}\bigg)\noisedistribution(v)dv \\
    & \geq  \int \bigg(\sum_{\qimdp\in\Qimdp}\min_{x\in \qimdp}V_{k-1}^{\strX}(x)\int_{\qimdp} t(\bar{x}\mid x,\strX(x+v))d\bar{x}\bigg)\noisedistribution(v)dv \\
    & \geq  \int \sum_{\qimdp\in\Qimdp}V_{k-1}^{\strX}(\qimdp) T(\qimdp\mid x,\strX(x+v))\noisedistribution(v)dv \\
    & \geq \min_{\aimdp \in \Aimdp} \sum_{\qimdp\in\Qimdp}\plow^{k-1}(\qimdp) T(\qimdp \mid x,a) \\
    & \geq \min_{\aimdp} \min_{\adv_{\qimdp}^{\aimdp}} \sum_{\qimdp'} \adv_{\qimdp}^{\aimdp}(\qimdp') \plow^{k-1}(\qimdp') \\
    & = \plow^{k}(q),
\end{align*}
where the last inequality follows by Theorem \ref{Th:TransitionBounds} 
while the second to last inequality follows by the induction assumption and by the observation that the noise distribution $\noisedistribution$ only affects the choice of the action. 
\end{proof}

Theorem~\ref{th:correctness} extends the \imdp verification guarantees to Process~\eqref{eq:system}, even though Process~\eqref{eq:system} is unknown \emph{a priori} and observed with noise.
We provide empirical validation of Theorem~\ref{th:correctness} in our case studies in Section~\ref{sec:studies} by comparing the verification results of a known and learned system.

\subsection{End-to-End Algorithm}
Here, we provide an overview of the entire verification framework in Algorithm~\ref{alg:verification} and its complexity.
The algorithm takes dataset $\dataset$, set $\FullSet\subset\reals^n$, and specification formula $\Spec$, and returns three sets of initial states $\Qyes$, $\Qno$, and $\Qposs$, from which Process~\eqref{eq:system} is guaranteed to satisfy, violate, and possibly satisfy $\Spec$ respectively.
Assuming a uniform discretization of $\FullSet$, the number of states in the abstraction is $|\Qimdp| = 2^n+1$.

First, GP regression is performed $n\,|A|$ times in Line \ref{line:learning} of Algorithm \ref{alg:verification}.
In general, each GP regression is $\mathcal{O}(\datasetsize^3)$, where $\datasetsize$ is the number of input points in the dataset.
GP regression can reach a complexity closer to $\mathcal{O}(d^{2.6})$ through improved Cholesky factorization, e.g.~\cite{camarero2018simple}, although in practice many GP regression toolboxes are fast enough.
As this is repeated for each action and output component, the total complexity of performing all regressions is $\mathcal{O}(\datasetsize^3 \, n \, |A|)$.

Next, the GPs are used to find the image and posterior covariance supremum of region in the discretization in Algorithm~\ref{alg:images}. 
We perform this computation with a branch-and-bound optimization procedure introduced in \cite{blaas2020adversarial} until one of two termination criteria are reached: a maximum search depth $T$, or a minimum distance between the bounds (e.g., $10^{-3}$).
In our case studies, we observe that the algorithm usually terminates long before the maximum search depth is reached.
The complexity of this operation is $\mathcal{O}(2^T \, \datasetsize^2 \, n \, |\Qimdp|\, |\Aimdp| )$, where $T$ allows a trade-off between accuracy and computation time.  
Hence, for a fixed $T$, the computations in Lines~\ref{line:image} and \ref{line:error} are polynomial in size of the dataset $\datasetsize$, dimensions $n$, abstraction $|\Qimdp|$, action set $|\Aimdp|$.  We note that 
the branch-and-bound method is exponential in $T$, and  $|\Qimdp|$ can be exponential in the size of the dimensions if, e.g., a uniform discretization is used.

The abstraction is completed by calculating the transition probability intervals between each pair of discrete states under each action using Theorem~\ref{Th:TransitionBounds} and Proposition~\ref{prop:distbound} (Lines~\ref{line:transitions}-\ref{line:TransitionBounds}) in Algorithm \ref{alg:transitions}.
The calculation of the transition probability intervals between every state-action state pair is $\mathcal{O}(|\Qimdp|^2 \, |\Aimdp|)=\mathcal{O}(2^{2n} \, |\Aimdp|)$, which is standard for IMDP construction.

Finally, we use an existing tool for verifying the abstraction \imdp against specification formula $\Spec$ with its inherent operator $\bowtie$ and threshold value $\p$ (Line~\ref{line:verify}). 
The correctness of the obtained results is guaranteed by Theorem~\ref{th:correctness}.
PCTL model checking of \imdp is polynomial in the size of the state-action pairs and length of the PCTL formula \cite{Lahijanian:TAC:2015}.  
Hence, the introduction of a learned component to abstraction-based verification does not fundamentally make the problem harder, as the exponential influence of the discretization of the state-space ($2^{2n}$) is still dominant.

\begin{algorithm}
   \caption{Bound State Images and Errors} 
   \label{alg:images}
   \hspace*{\algorithmicindent} \textbf{Input:} Posterior GPs $\{\mathcal{GP}^\aimdp\}$, actions $A$, IMDP states $\Qimdp$ \\
    \hspace*{\algorithmicindent} \textbf{Output:} Image and error bounds $Im$, $e$ 
   \begin{algorithmic}[1]
        \State $Im\gets \{\},e\gets \{\}$
        \For{$\qimdp \in \Qimdp,\aimdp\in \Aimdp$}
    
        \State $Im(\qimdp,\aimdp)\gets$ OverapproximateImage$(\qimdp, \mathcal{GP}^\aimdp)$\label{line:image}
        \State $e(\qimdp,\aimdp)\gets$ BoundError$(\qimdp, \mathcal{GP}^\aimdp)$ \hfill \label{line:error}
    \EndFor
    \Return $Im$, $e$
   \end{algorithmic}
\end{algorithm}

\begin{algorithm}
    \caption{Calculate Transition Intervals}
    \label{alg:transitions}
    \hspace*{\algorithmicindent} \textbf{Input:} Image and error bounds $Im$, $e$ , actions $A$, IMDP states $\Qimdp$ \\
    \hspace*{\algorithmicindent} \textbf{Output:} Transition interval matrices $\Plow$, $\Pup$ 
    \begin{algorithmic}[1]
    \State $\Plow\gets\{\}, \Pup\gets\{\}$
        \For{$(\qimdp,\qimdp')\in \Qimdp\times \Qimdp,a\in A$\label{line:transitions}} 
        \State $\distBound\gets$ Proposition~\ref{prop:distbound}
        \State $\Plow(\qimdp,\aimdp,\qimdp'),\Pup(\qimdp,\aimdp,\qimdp')\gets$ Theorem~\ref{Th:TransitionBounds} using $Im,e$ \label{line:TransitionBounds}
   \EndFor 
   \Return $\Plow,\Pup$
    \end{algorithmic}
\end{algorithm}

\begin{algorithm}
 \caption{Learning, Abstraction and Verification}\label{alg:verification}
\hspace*{\algorithmicindent} \textbf{Input:} Data $\dataset$, set $\FullSet$, spec. $\Spec$, $\APs$, label fcn. $L$ \\
\hspace*{\algorithmicindent} \textbf{Output:} Initial states $\FullSet_0$ that satisfy $\Spec$
\begin{algorithmic}[1]
    \State $\{\mathcal{GP}^\aimdp\}\gets$ GP regressions for each $\aimdp$ using $\dataset$ \label{line:learning}
    \State $\Qimdp\gets $Partition$(\FullSet) \cup\{\qimdp_{\Unsafe}\} $\label{line:discretize}
    
    \State $Im,e\gets$ Algorithm \ref{alg:images}
    \State $\Plow,\Pup\gets$ Algorithm \ref{alg:transitions}
     \State $\I\gets$AssembleIMDP$(\Qimdp, \Aimdp, \Plow, \Pup, \APs, L)$ 
     \State $\Qyes, \Qno, \Qposs\gets$VerifyIMDP$(\I,\Spec)$ \label{line:verify}
\end{algorithmic}
\end{algorithm}

\begin{figure*}[t]
    \centering
    \newcommand\figwidth{0.20\textwidth}
    \newcommand\vertspace{-7mm}
    \begin{subfigure}{0.7\textwidth}
    \centering
    \includegraphics[width=0.35\linewidth]{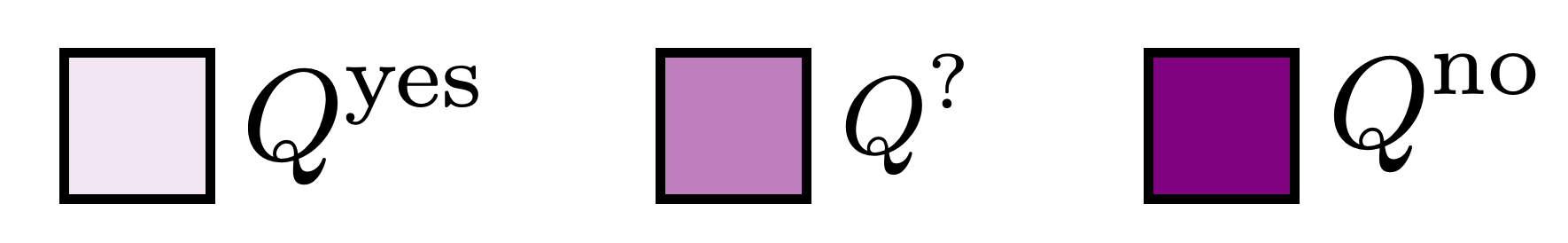}
    \end{subfigure}\\
    \begin{subfigure}{\figwidth}
    \includegraphics[width=\linewidth, trim={1.5cm, 1.5cm, 0.75cm, 0.75cm}, clip]{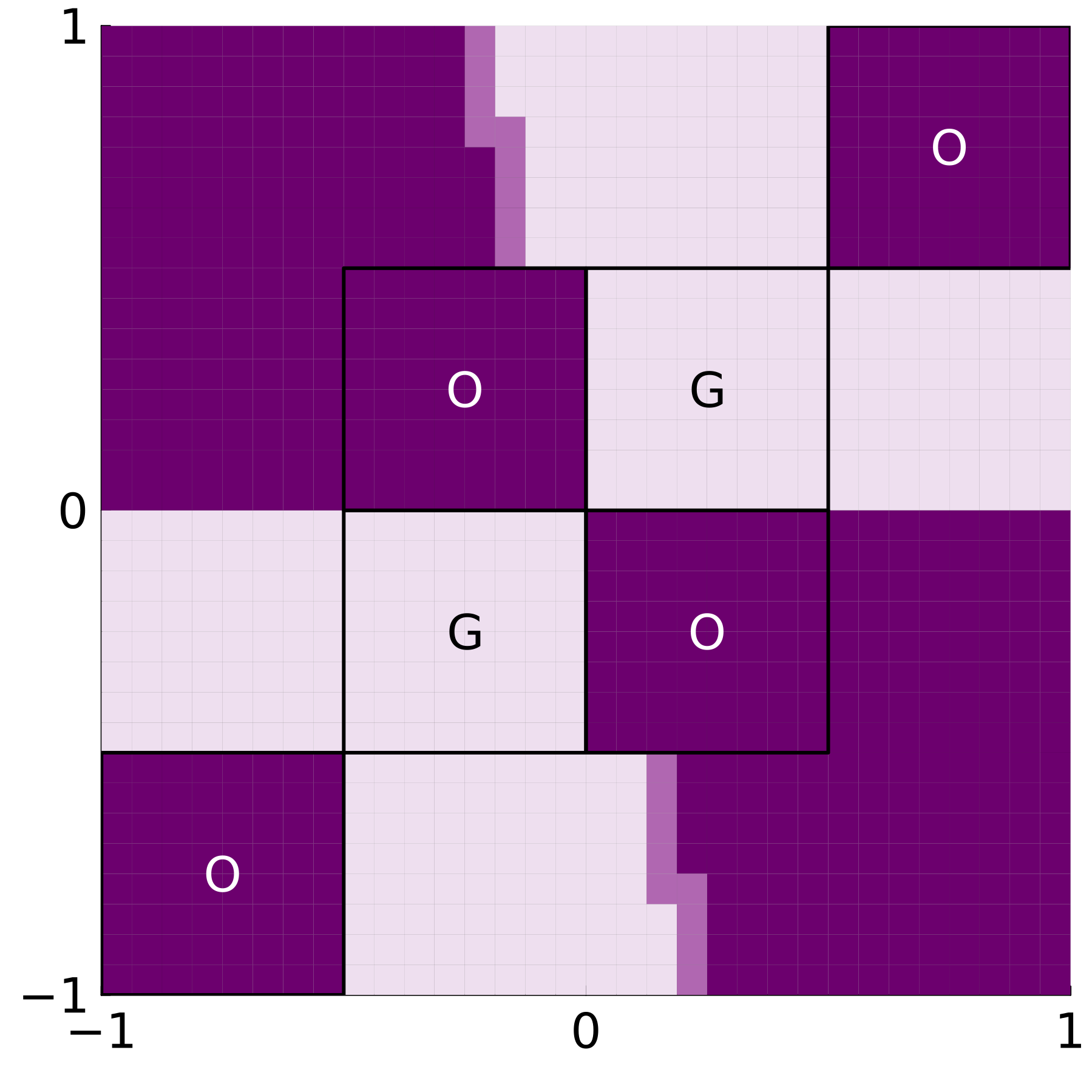}
    \vspace{\vertspace}
    \caption{Known System}
    \label{fig:true-boundeduntil-1step}
    \end{subfigure}
    \begin{subfigure}{\figwidth}
    \includegraphics[width=\linewidth]{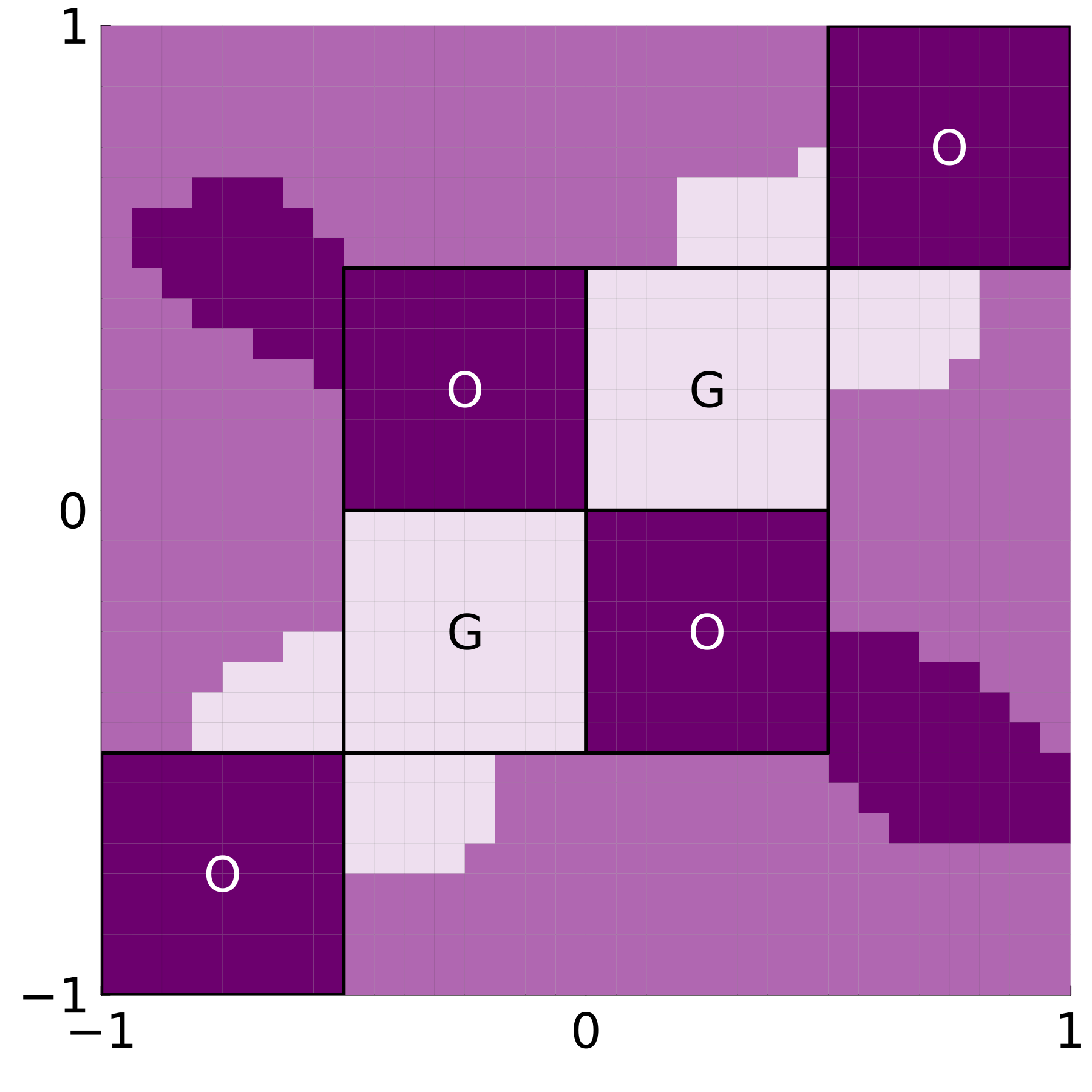}
    \vspace{\vertspace}
    \caption{100 data points}
    \label{fig:example100-boundeduntil-1step}
    \end{subfigure}
    \begin{subfigure}{\figwidth}
    \includegraphics[width=\linewidth]{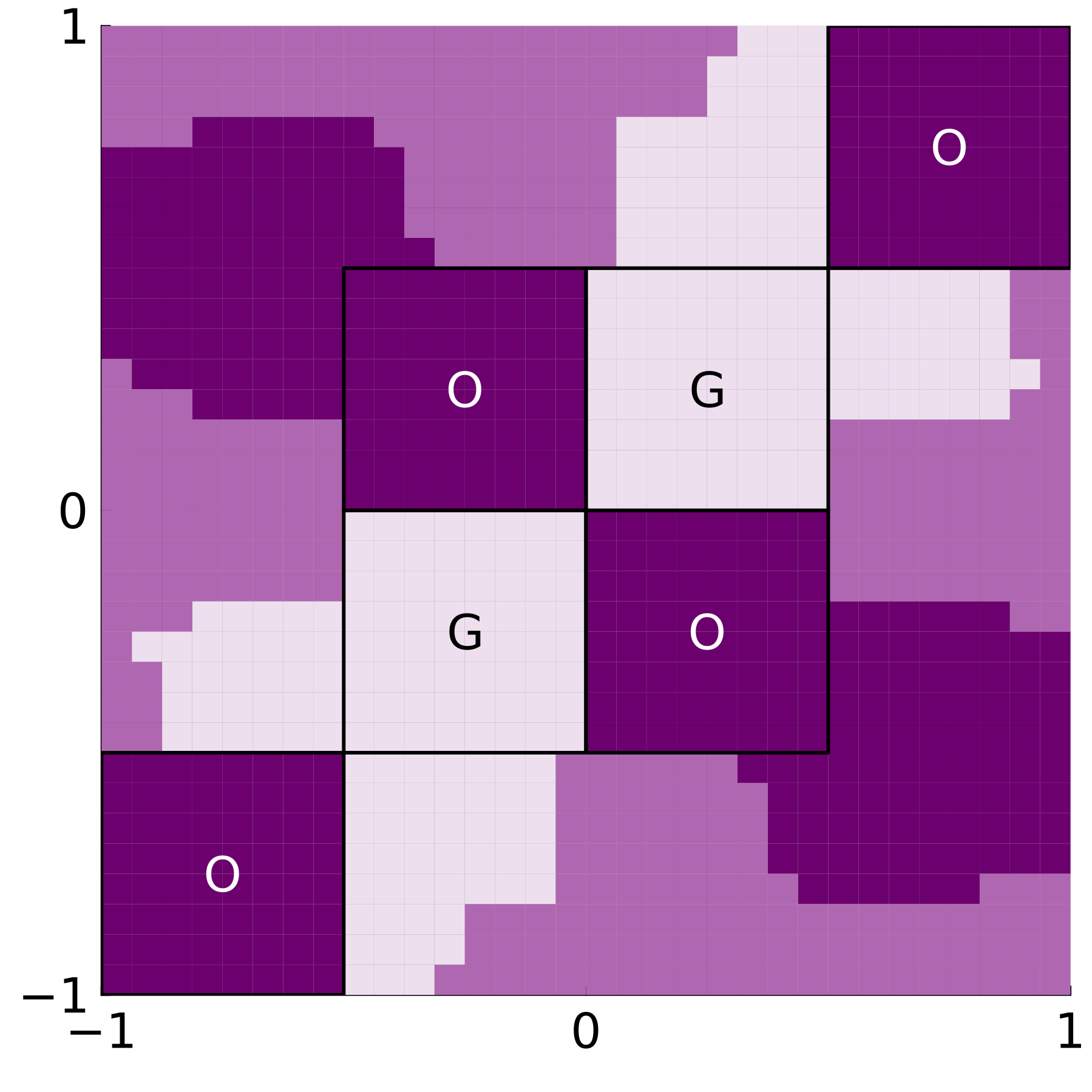}
    \vspace{\vertspace}
    \caption{500 data points}
    \label{fig:example500-boundeduntil-1step}
    \end{subfigure}
    \begin{subfigure}{\figwidth}
    \includegraphics[width=\linewidth]{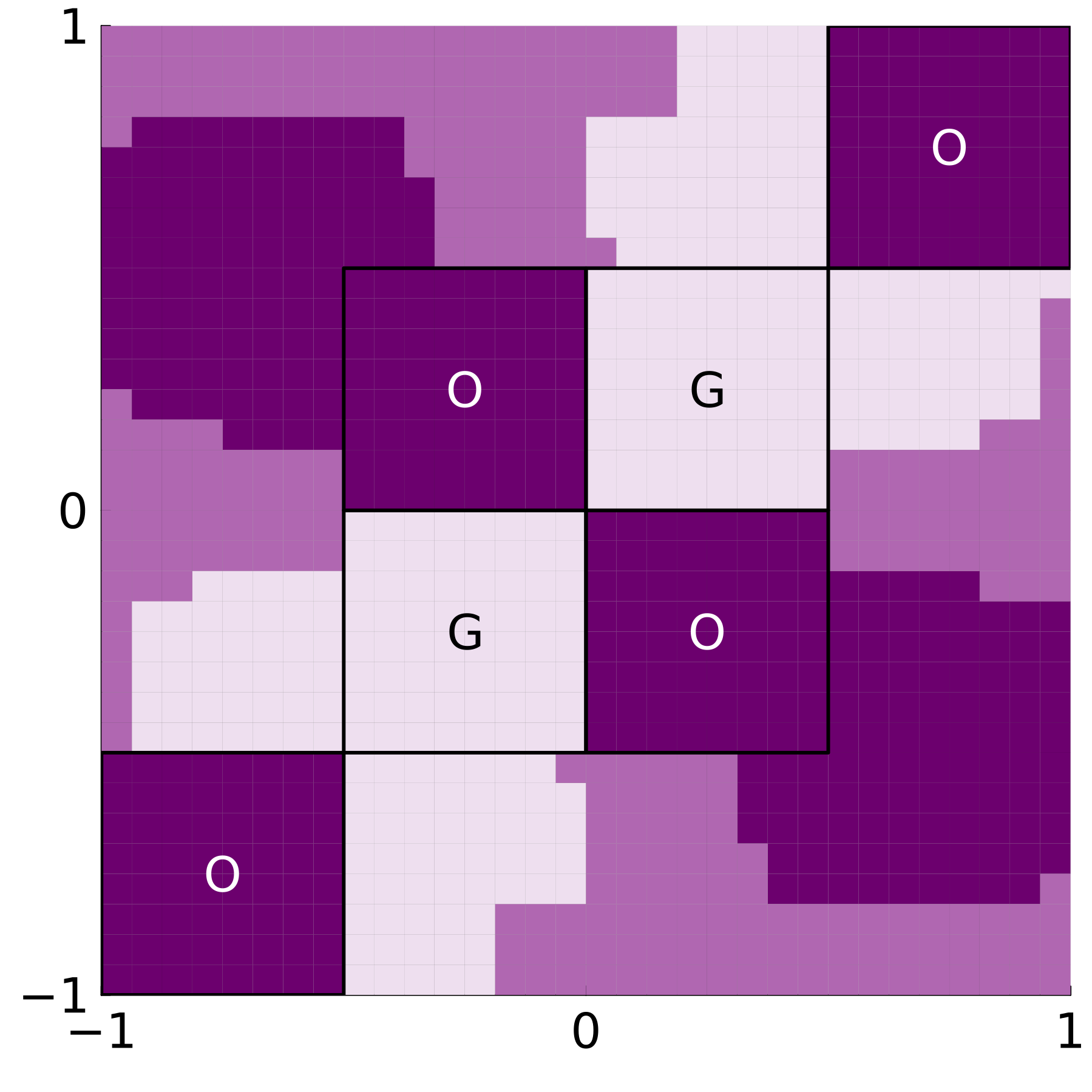}
    \vspace{\vertspace}
    \caption{2000 data points}
    \label{fig:example2000-boundeduntil-1step}
    \end{subfigure}
    \caption{The unbounded ($k=\infty$) verification results for the known system and learned system using increasingly large datasets subject to $\Spec=\Pr_{\geq0.95}[\neg O~\UntilOp~D$].
    }
    \label{fig:boundeduntil-verification}
\end{figure*}

\section{Case Studies}
\label{sec:studies}

We demonstrate the efficacy of our data-driven verification framework through several case studies. 
We begin with verifying a single-mode linear system and compare it with a barrier function approach. 
Then, we show the influence of different parameters on the quality of the verification results and computation times. 
Finally, we demonstrate the efficacy of our approach on switching and nonlinear systems. 

For each demonstration, we consider compact state space $\FullSet$ and use a uniform-grid discretization of $\FullSet$ of cells with side length $\DiscretizationParam$.
The synthetic datasets are generated by uniformly sampling states from $\FullSet$, propagating through $f$, and adding Gaussian noise with zero mean and $0.01$ standard deviation.
The squared-exponential kernel with length scale $1$ and scale factor $1$ and the zero mean function are used as the priors for GP regression, and we selected an upper-bound for the numerator for each example to employ Proposition~\ref{eq:RKHS_norm}.
An implementation of our tool is available online\footnote{\url{https://github.com/jmskov/TransitionIntervals.jl}}.

\subsection{Linear System}
First we perform a verification of an unknown linear system with dynamics 
$$
\x(k+1) = \begin{bmatrix}
    0.4 & 0.1\\
    0.0 & 0.5
    \end{bmatrix} \x(k), \quad \y(k) = \x(k) + \noise(k)
$$
and regions of interest with labels $D$ and $O$ as shown in Fig.~\ref{fig:boundeduntil-verification}.
The PCTL formula is $\Spec=\Pr_{\geq0.95} [\neg O~\UntilOp~D]$, which states that a path initializing at an initial state does not visit $O$ until $D$ is reached with a probability of at least 95\%. 
The verification result using the known dynamics is shown in Fig.~\ref{fig:true-boundeduntil-1step}, which provides a basis for comparing the results of the learning-based approach.
This baseline shows the initial states that belong to $\Qyes$, $\Qno$, and the indeterminate set $\Qposs$.
Some states belong to $\Qposs$ due to the discretization, and subsets of $\Qposs$ either satisfy or violate $\Spec$.
Fig.~\ref{fig:true-boundeduntil-1step} asserts that most discrete states ideally belong to either $\Qyes$ or $\Qno$.

The verification results of the unknown system with various dataset sizes are shown in Fig.~\ref{fig:example100-boundeduntil-1step} to Fig.~\ref{fig:example2000-boundeduntil-1step}.
With 100 data points, the learning error dominates the quality of the verification result as compared to the baseline.
Nevertheless, even with such a small dataset (used to learn $f$), the framework is able to identify a subset of states belonging to $\Qyes$ and $\Qno$ as shown in Fig.~\ref{fig:example100-boundeduntil-1step}.  
Figs.~\ref{fig:example500-boundeduntil-1step} and \ref{fig:example2000-boundeduntil-1step} show these regions growing when using 500 and 2000 datapoints respectively. 
The set of indeterminate states, $\Qposs$, begins to converge to the baseline result but the rate at which it does slows with larger datasets.
The computation times for the datasets with 100, 500, and 2000 points were 2.9, 10.3, and 140 seconds, respectively.

Compared to the abstraction built using the known system, the learning-based abstractions are 48, 24, and 16 times more conservative, respectively, in terms of average transition interval widths. In the limit of infinite data, the error could be driven to zero even with imperfect RKHS constant approximations, which would cause the learning-based results to converge to the known results in Fig.~\ref{fig:true-boundeduntil-1step}.
Additional uncertainty reduction depends on refining the states in $\Qposs$ through further discretization at the expense of additional memory and computational time.
Below, we examine the trade-offs between discretization resolution with the quality of the final results and the effect on total computation time.

\subsubsection{Comparison with Barrier functions}

We compare our approach with a barrier method based on Sum-of-Squares (SOS) polynomials with barrier degree eight~\cite{jagtap2020formal,wajid2022formal}.
The verification task is to remain within the state $[-1,0.5]^2$ while avoiding the region with label $O$ over one then five time steps with a probability at least 95\%.
The key limitations of the barrier method are threefold.
First, barrier functions are generated with respect to an initial set, so the generation of the barrier is repeated for every discrete state. 
Second, barrier functions are unable to generate meaningful results when the initial state is adjacent to an unsafe state due to existing methods finding only smooth barriers (such as a polynomial function).
Finally, the barrier methods use the uncertainty from Proposition~\ref{prop:distbound} as a confidence wrapper, meaning the results hold with a confidence.  In this case study, we set the confidence to $95\%$.

The methods are compared first using the true system, and then using 500 datapoints to learn the system.
The state color in Figure \ref{fig:comparison} indicates satisfaction using our approach, and the state outlines indicate those using the barrier method.
Our approach provided more safety guarantees in both the known and data-driven settings, especially beyond one step. 
After five steps, the barrier-based method cannot guarantee satisfaction anywhere.
The complete runtimes for the learned results are 0.82 seconds for our approach and 12 seconds for the barrier method.
Our method overcomes each of the limitations of the barrier approach, while running ten times faster.

\begin{figure}[t]
    \centering
    \newcommand\figwidth{0.24\textwidth}
    \newcommand\vertspace{-7mm}
    \begin{subfigure}{\linewidth}
    \centering
    \includegraphics[width=0.35\linewidth]{images/legend.png}
    \end{subfigure}\\

    \begin{subfigure}{\figwidth}
    \includegraphics[width=\linewidth]{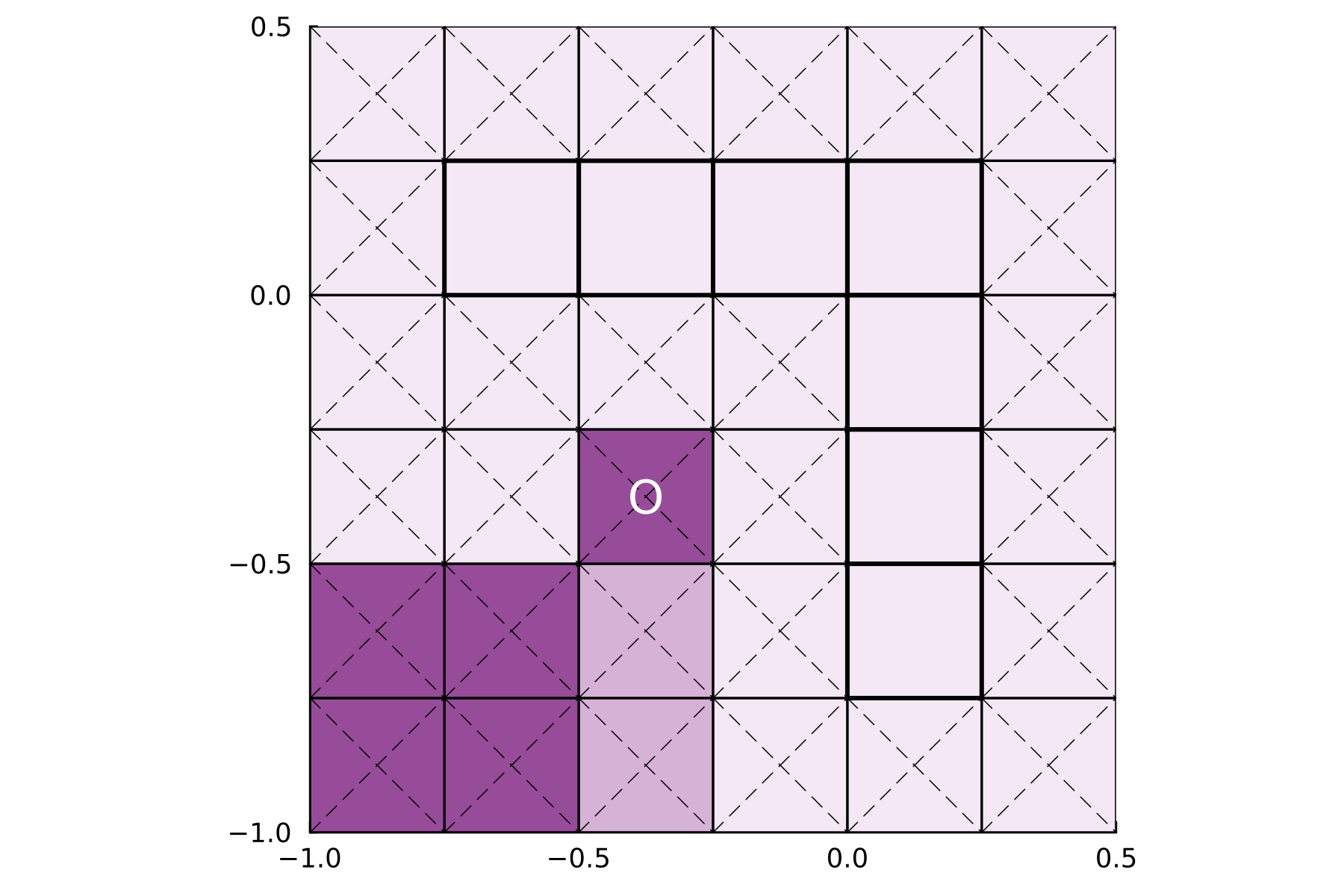}
    \vspace{\vertspace}
    \caption{Known System, $k=1$}
    \label{fig:comp-known}
    \end{subfigure}
    \begin{subfigure}{\figwidth}
        \includegraphics[width=\linewidth]{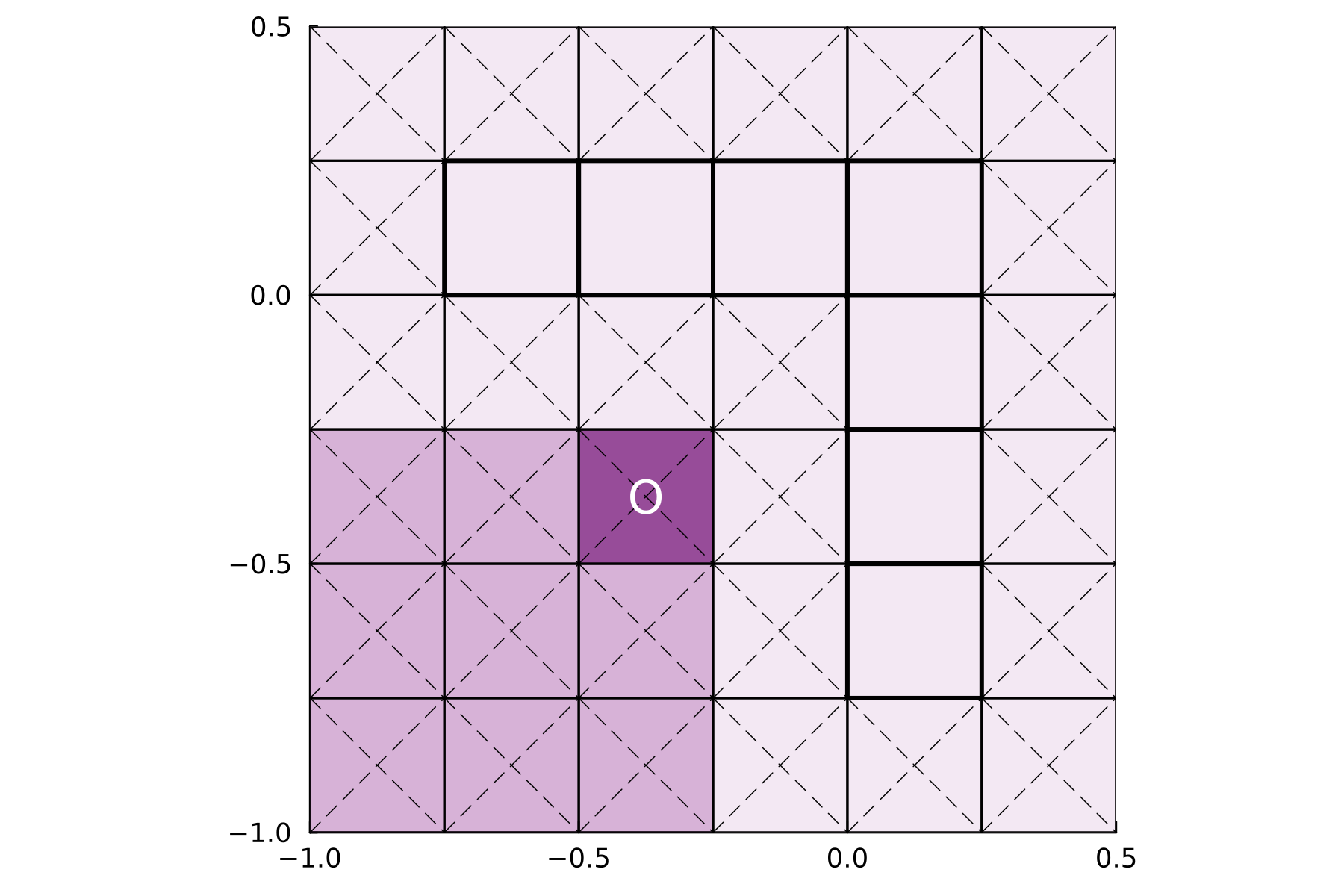}
        \vspace{\vertspace}
        \caption{Learned System, $k=1$}
        \label{fig:comp-data}
        \end{subfigure}
    \begin{subfigure}{\figwidth}
        \includegraphics[width=\linewidth]{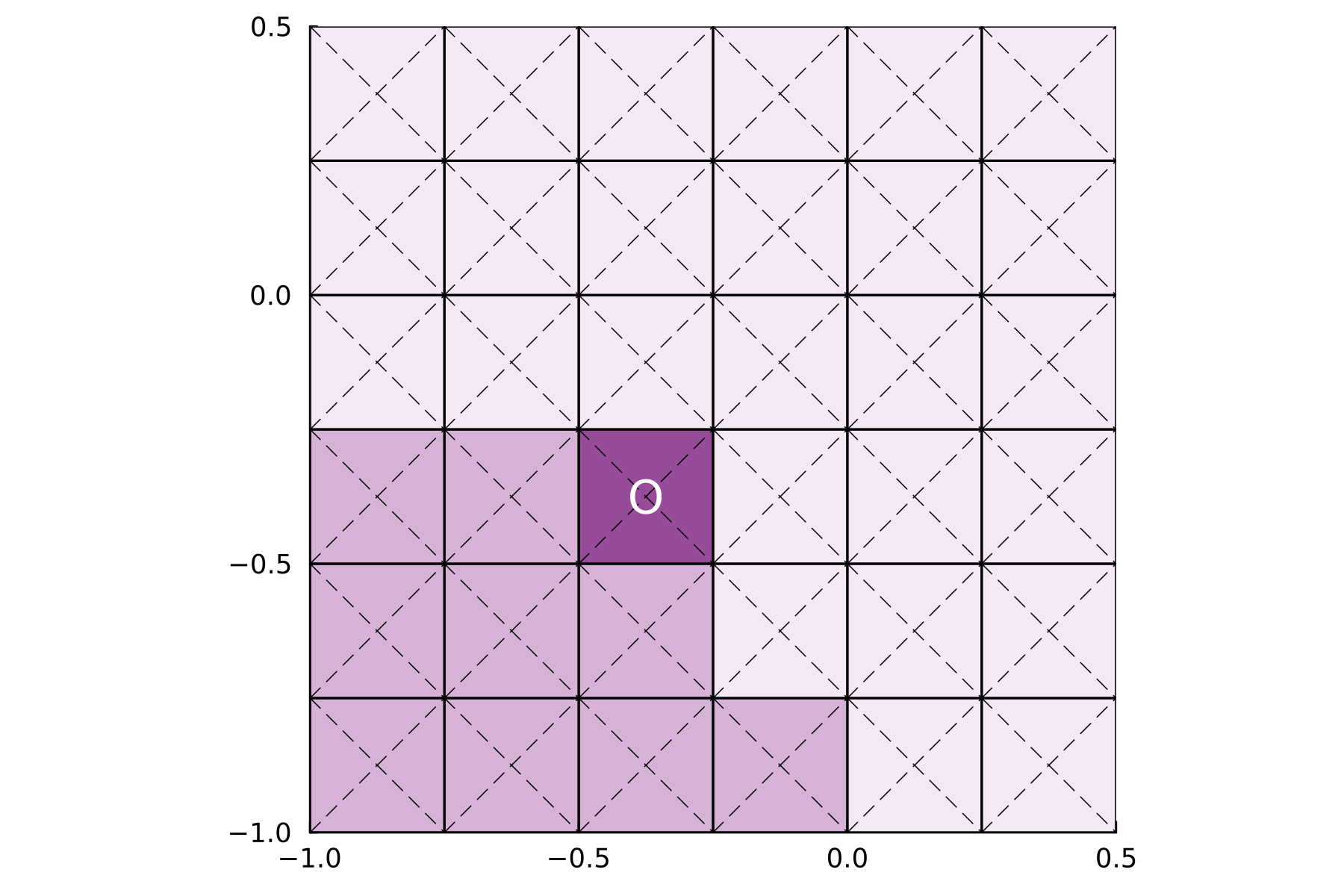}
        \vspace{\vertspace}
        \caption{Learned System, $k=5$}
        \label{fig:comp-data-5}
        \end{subfigure}
    \caption{Verification using our approach (colors) and SOS barriers (outlines) for the known and learned systems.
    Bold outlines indicate a satisfying barrier is found for that state, while dashed Xs indicate the barrier does not satisfy safety.}
    \label{fig:comparison}
\end{figure}

\subsection{Parameter Considerations}
The verification results on a linear system and computation times are compared for varying 
the learning error and discretization parameters, $\distBound$ and $\Delta$ respectively,
and the efficacy of using Proposition~\ref{prop:distbound} to choose $\distBound$ is highlighted.
The unknown linear system is given by 
$$
\x(k+1) = \begin{bmatrix}
    0.8 & 0.5\\
    0.0 & 0.5
    \end{bmatrix} \x(k), \quad \y(k) = \x(k) + \noise(k).
$$
The specification is the safety formula $\Pr_{\geq0.95} [\GloballyOp~ \FullSet]$, i.e., the probability of remaining within $\FullSet$ is at least 95\%. 

We used 200 datapoints for the verification of this system.
To compare the effect of $\distBound$, we manually selected values to be applied for every transition and compared the results with using the criterion in Proposition~\ref{prop:distbound}.
The discretization fineness is studied by varying $\Delta$ to produce abstractions consisting of 17 states up to 10001 states. 

Fig.~\ref{fig:discretization-epsilon-comparison} compares the average satisfaction probability interval sizes defined by 
$$
\bar{p} = \frac{1}{|\Qimdp|}\sum_{\qimdp\in\Qimdp}\pup(\qimdp) - \plow(\qimdp)
$$
where a smaller value indicates more certain (less-conservative) intervals, i.e., 
upper bound and lower bound probabilities are similar.
As the uniform choice of $\distBound$ approaches 0.15, the average probability interval size decreases although this trend reverses as $\distBound$ decreases further. 
This is due to $\distBound$ becoming too small to get non-zero regression error probabilities from Proposition~\ref{th:RKHS}. 
However, by applying Proposition~\ref{prop:distbound}, we obtain an optimal value for $\distBound$ for every region pair, resulting in the smallest interval averages $\bar{p}$ for every discretization.

Note that there is an asymptotic decay in the average probability interval size as the discretization gets finer.
This however comes at a higher computation cost.
The total computation times for framework components are shown in Table~\ref{tab:discretization-effect}.
The most demanding component involves discretizing $\FullSet$ and determining the image over-approximation and error upper bound of each state as discussed in Section~\ref{subsec:tbounds}. 
The transition probability interval calculations involve checking the intersections between the image over-approximations and their respective target states.
The unbounded verification procedure is fast relative to the time it takes to construct the \imdp abstraction. 

\begin{figure}
    \newcommand{\figwidth}{0.45\textwidth}
    \centering
    \includegraphics[width=\figwidth]{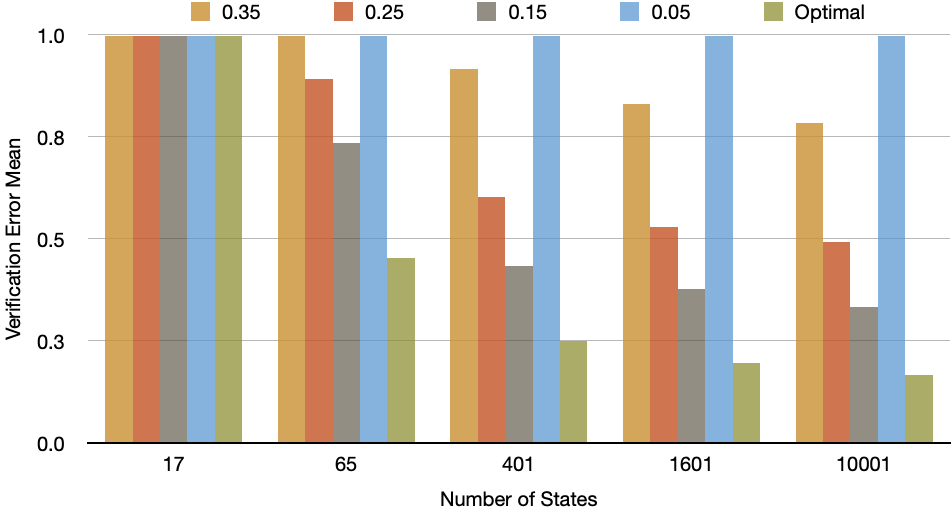}
    \caption{Effect of the size of discretization and $\distBound$ on the average satisfaction probability interval size (verification error).}
    \label{fig:discretization-epsilon-comparison}
\end{figure}

\begin{table}[t]
\centering
\caption{Component computation times for an increasing number of states. Components include discretization and image over-approximations, transition probability interval calculations, and final verification times in seconds.}
\label{tab:discretization-effect}
    \begin{tabular}{
        @{}
        l  
        S
        S
        S
        @{}
        }
        \toprule
        \multirow{2}{*}{States} & \multicolumn{3}{c}{\underline{ \hspace{1.9cm} Computation Time (s) \hspace{1.6cm}}} \\
         & {Disc. + Images} & {\quad Transitions} & {\quad Verification} \\
        \hline
        65 & 0.202 & 0.000691 & 0.00104\\
        401 & 1.3 & 0.0191 & 0.000358\\
        1,601 & 5.39 & 0.309 & 0.00127\\
        10,001 & 32.3 & 12.2 & 0.00844\\
        40,001 & 135.0 & 204.0 & 0.0858\\
        62,501 & 208.0 & 516.0 & 0.151\\
        \bottomrule
    \end{tabular}
\end{table}

\subsection{Switched System}
We extend the previous example to demonstrate the verification of a system with multiple actions.
The unknown system $f(\x(k),a_i) = A_i\x(k)$ has two actions $\ControlSet=\{a_1,a_2\}$ where
\begin{align*}
    A_1 = \begin{bmatrix}
    0.8 & 0.5\\
    0.0 & 0.5
    \end{bmatrix},~~
    A_2 = \begin{bmatrix}
    0.5 & 0\\ 
    -0.5 & 0.8
    \end{bmatrix}.
\end{align*}
Verification was performed subject to $\Spec=\Pr_{\geq0.95} [\GloballyOp^{\leq k}~\FullSet]$ (the probability of remaining in $\FullSet$ within the $k$ time steps is at least 95\%) using $400$ datapoints for each action to construct the GPs.
The $k=1$ and $k=\infty$ results, in Figs.~\ref{fig:switched-1-step} and \ref{fig:switched-inf-step} respectively, show the regions where the system satisfies $\Spec$ regardless of which action is selected. 
The framework can handle a system with an arbitrary number of actions so long there is data available for each action.

\begin{figure}[tp]
    \centering
    \newcommand\figwidth{0.20\textwidth}
    \newcommand\vertspace{-7mm}
    \begin{subfigure}{\linewidth}
        \centering
        \includegraphics[width=0.35\linewidth]{images/legend.png}
        \end{subfigure}
    \begin{subfigure}{\figwidth}
    \includegraphics[width=\linewidth]{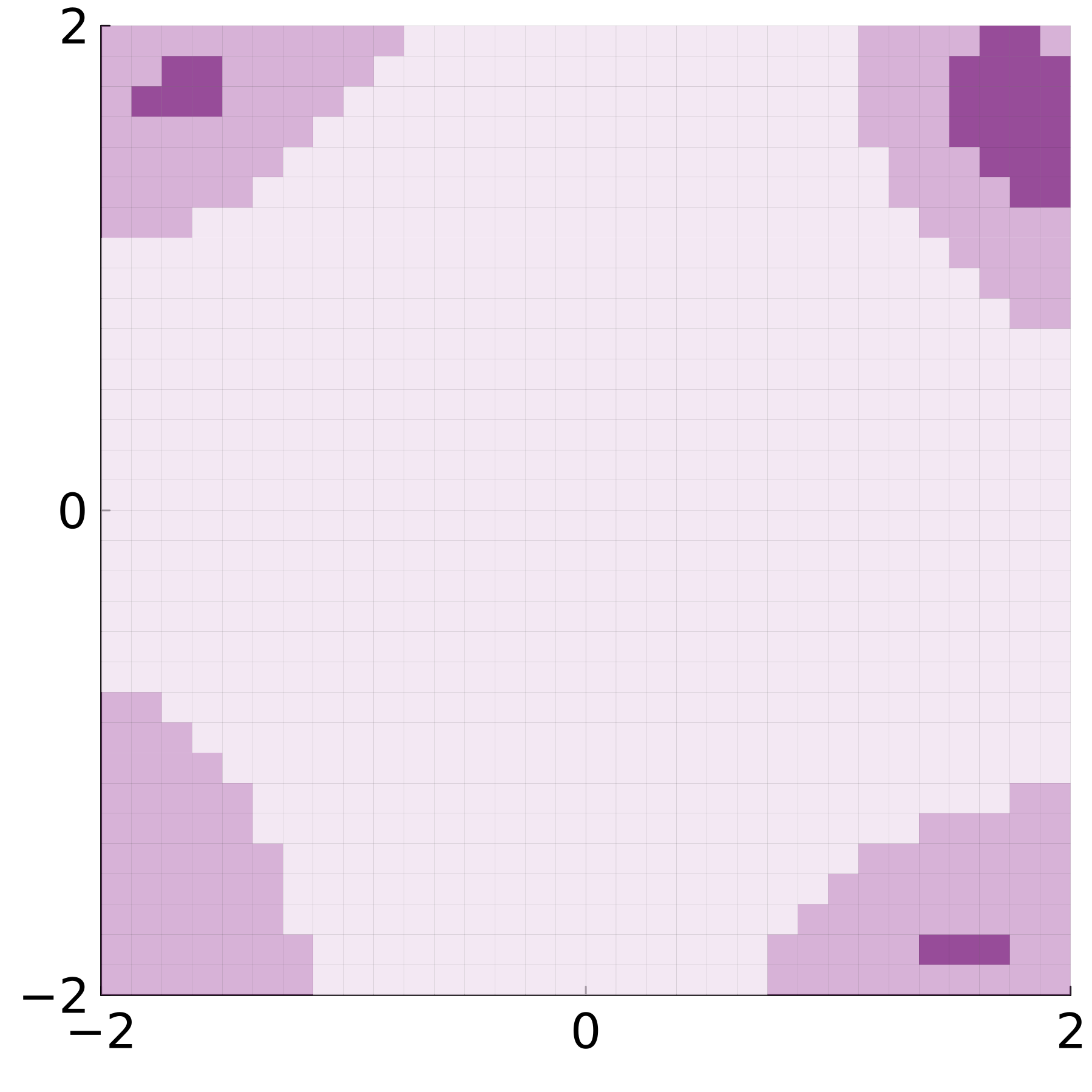}
    \vspace{\vertspace}
    \caption{$k=1$}
    \label{fig:switched-1-step}
    \end{subfigure}
    \begin{subfigure}{\figwidth}
    \includegraphics[width=\linewidth]{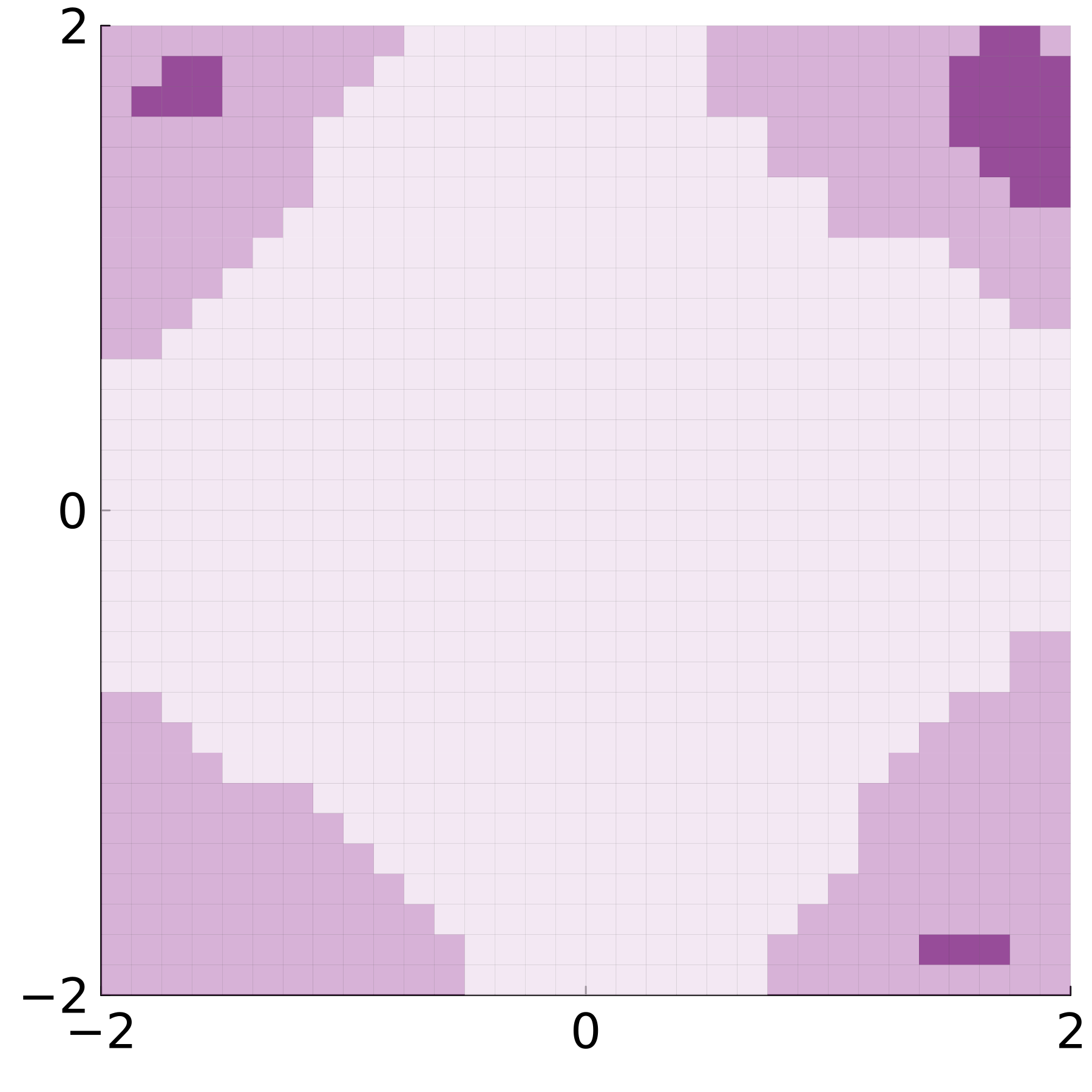}
    \vspace{\vertspace}
    \caption{$k=\infty$}
    \label{fig:switched-inf-step}
    \end{subfigure}
    \caption{Verification results for the switched system subject to $\Spec=\Pr_{\geq0.95} [\GloballyOp^{\leq k} \FullSet]$.
    }
    \label{fig:switched}
\end{figure}

\subsection{Autonomous Nonlinear System}
We perform verification on the unknown nonlinear system 
$$
f(\x(k))
= \x(k) + 0.2
\begin{bmatrix}
\x^{(2)}(k) + (\x^{(2)}(k))^2e^{\x^{(1)}(k)}\\  \x^{(1)}(k)
\end{bmatrix}
$$
with the vector field shown in Fig.~\ref{fig:nl5-true-vector-field}. 
The system is unstable about its equilibrium points at $(0,0)$ and $(0,-1)$ and slows as it approaches them. 
The flow enters the set $\FullSet$ in the upper-left
quadrant, and exits in the others. 
This nonlinear system could represent a closed-loop control system 
with dangerous operating conditions near the equilibria that should be avoided.

\begin{figure}[b]
    \centering
    \newcommand\figwidth{0.20\textwidth}
    \newcommand\vertspace{-7mm}
    \begin{subfigure}{\figwidth}
        \includegraphics[width=\linewidth,trim={1.5cm, 1.5cm, 0.75cm, 1.5cm}, clip]{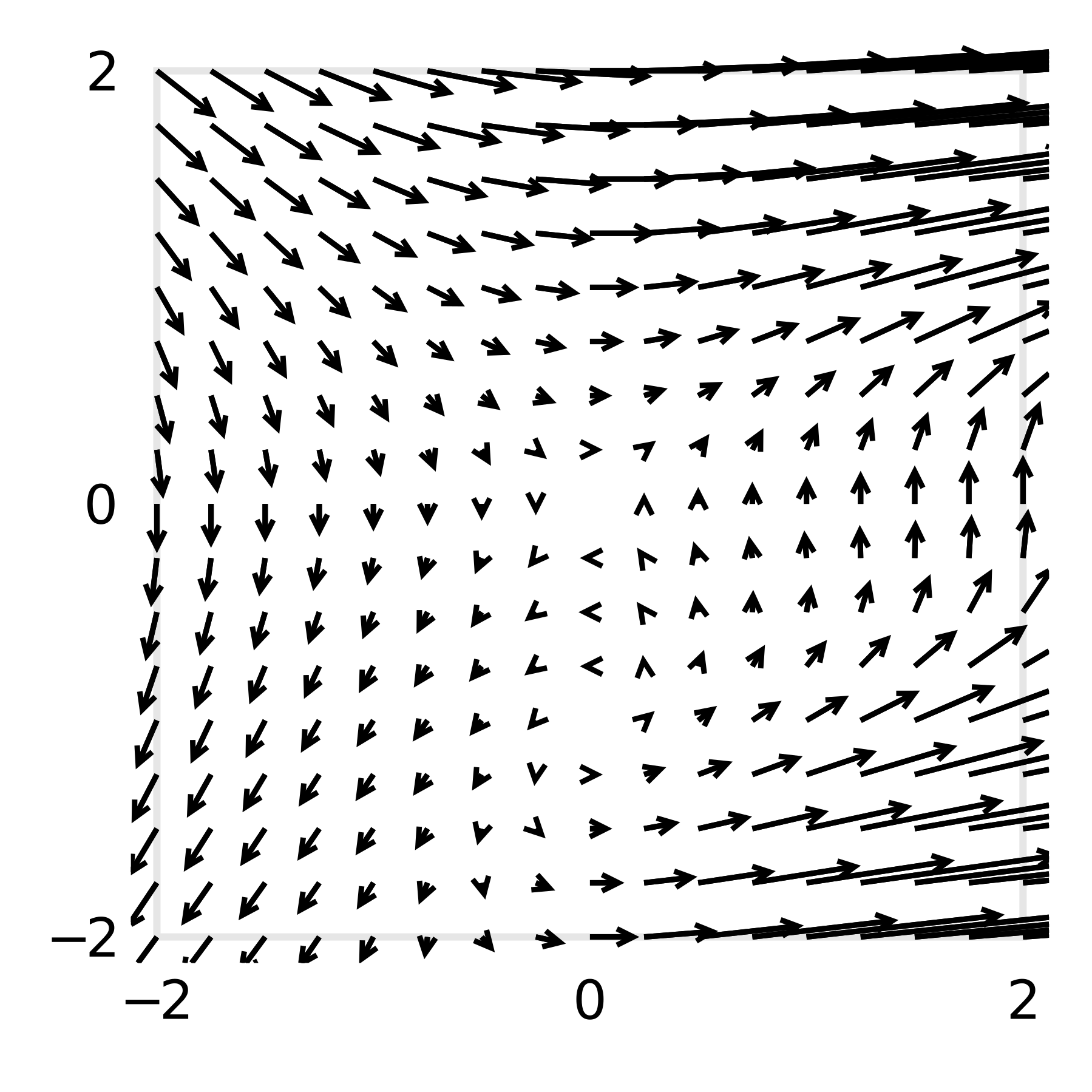}
        \vspace{\vertspace}
        \caption{True vector field}
        \label{fig:nl5-true-vector-field}
    \end{subfigure}
    \begin{subfigure}{\figwidth}
        \includegraphics[width=\linewidth,trim={1.5cm, 1.5cm, 0.75cm, 1.5cm}, clip]{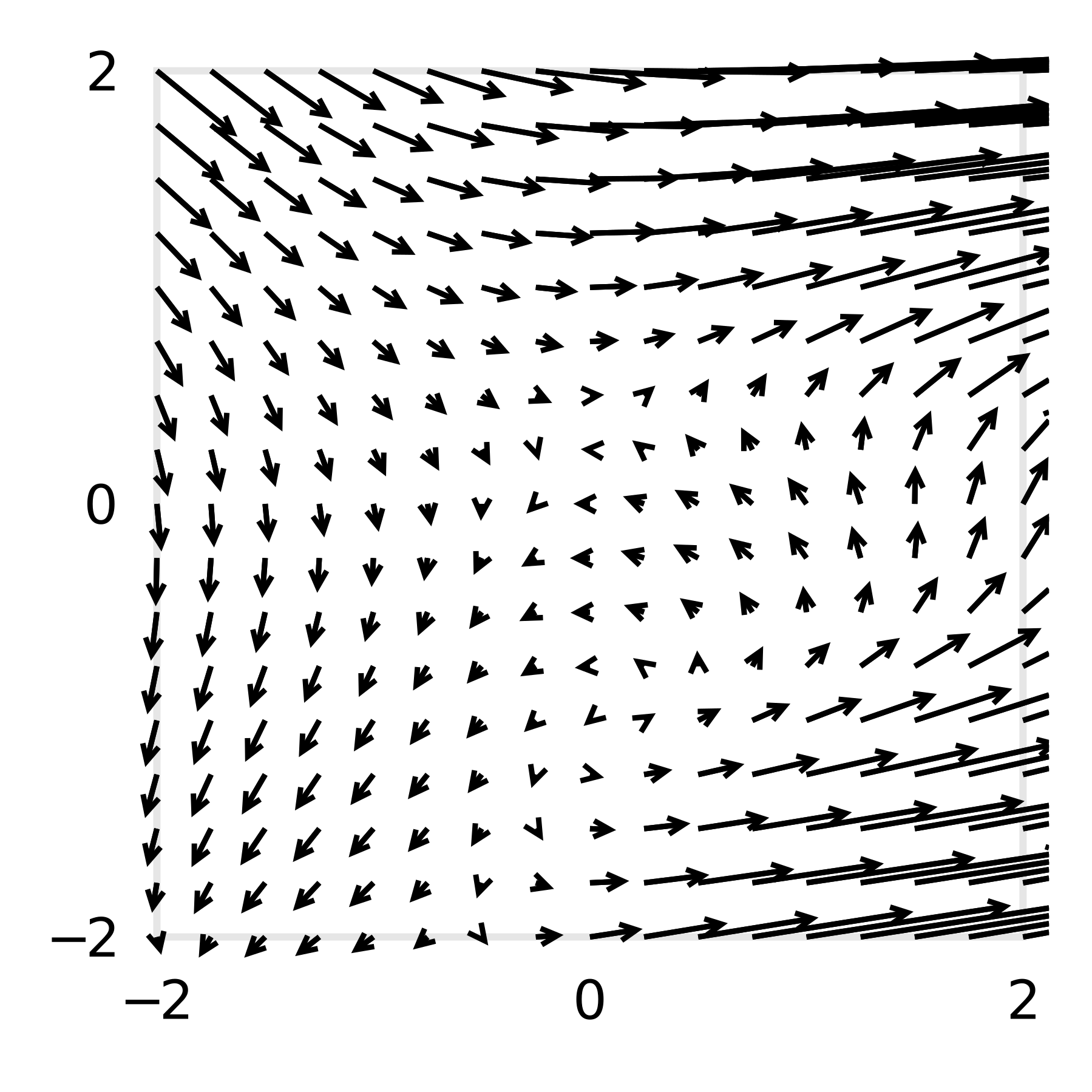}
        \vspace{\vertspace}
        \caption{Learned vector field}
        \label{fig:nl5-learned}
    \end{subfigure}
    \caption{Learned vector field for the autonomous nonlinear system.
    }
    \label{fig:nl5 vectorfield}
\end{figure}

\begin{figure}[t]
    \centering
    \newcommand\figwidth{0.20\textwidth}
    \newcommand\vertspace{-7mm}
    \begin{subfigure}{0.5\textwidth}
    \centering
    \includegraphics[width=0.35\linewidth]{images/legend.png}
    \end{subfigure}\\
    \begin{subfigure}{\figwidth}
    \includegraphics[width=\linewidth]{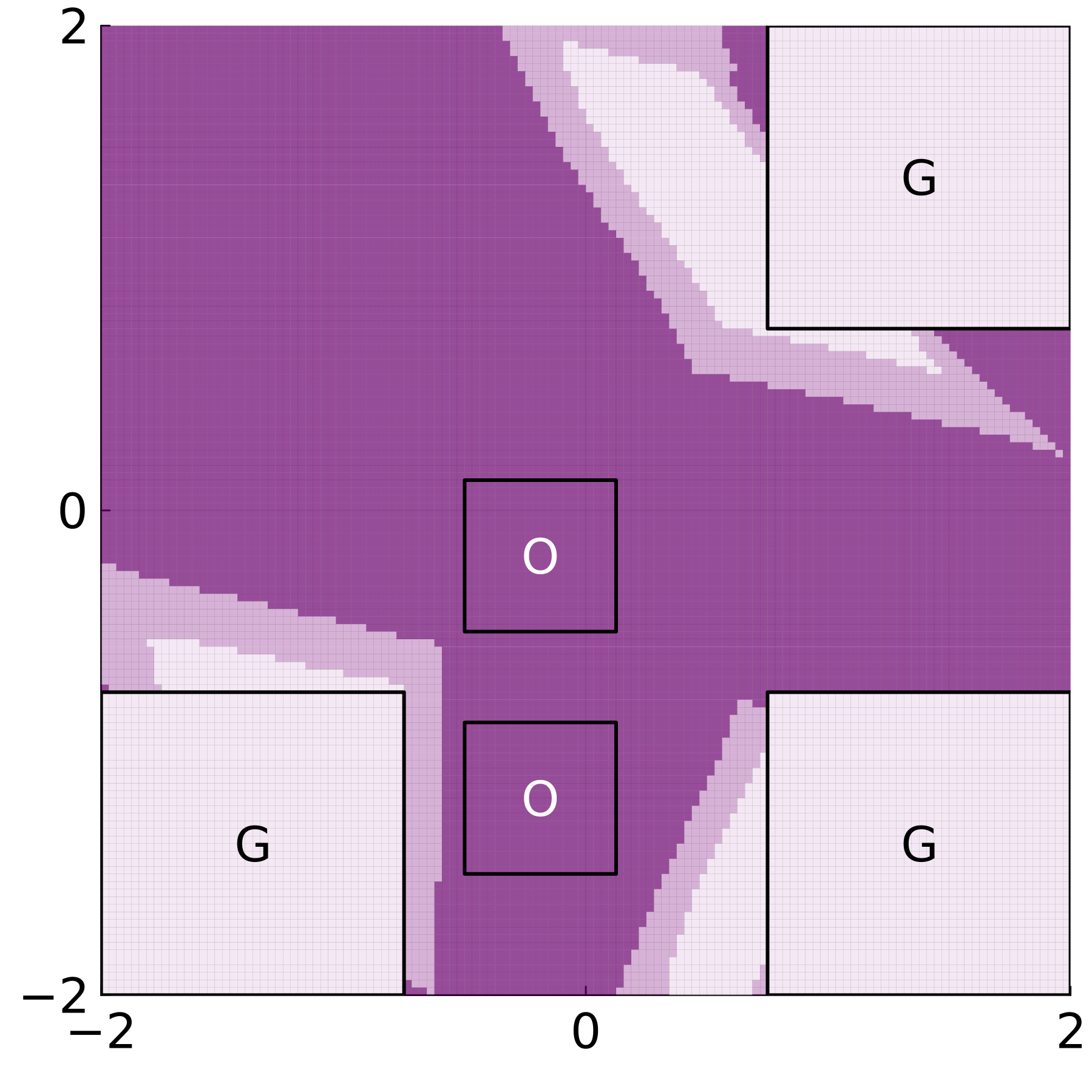}
    \vspace{\vertspace}
    \caption{$k=1$}
    \label{fig:nl5-1-step}
    \end{subfigure}
    \begin{subfigure}{\figwidth}
    \includegraphics[width=\linewidth]{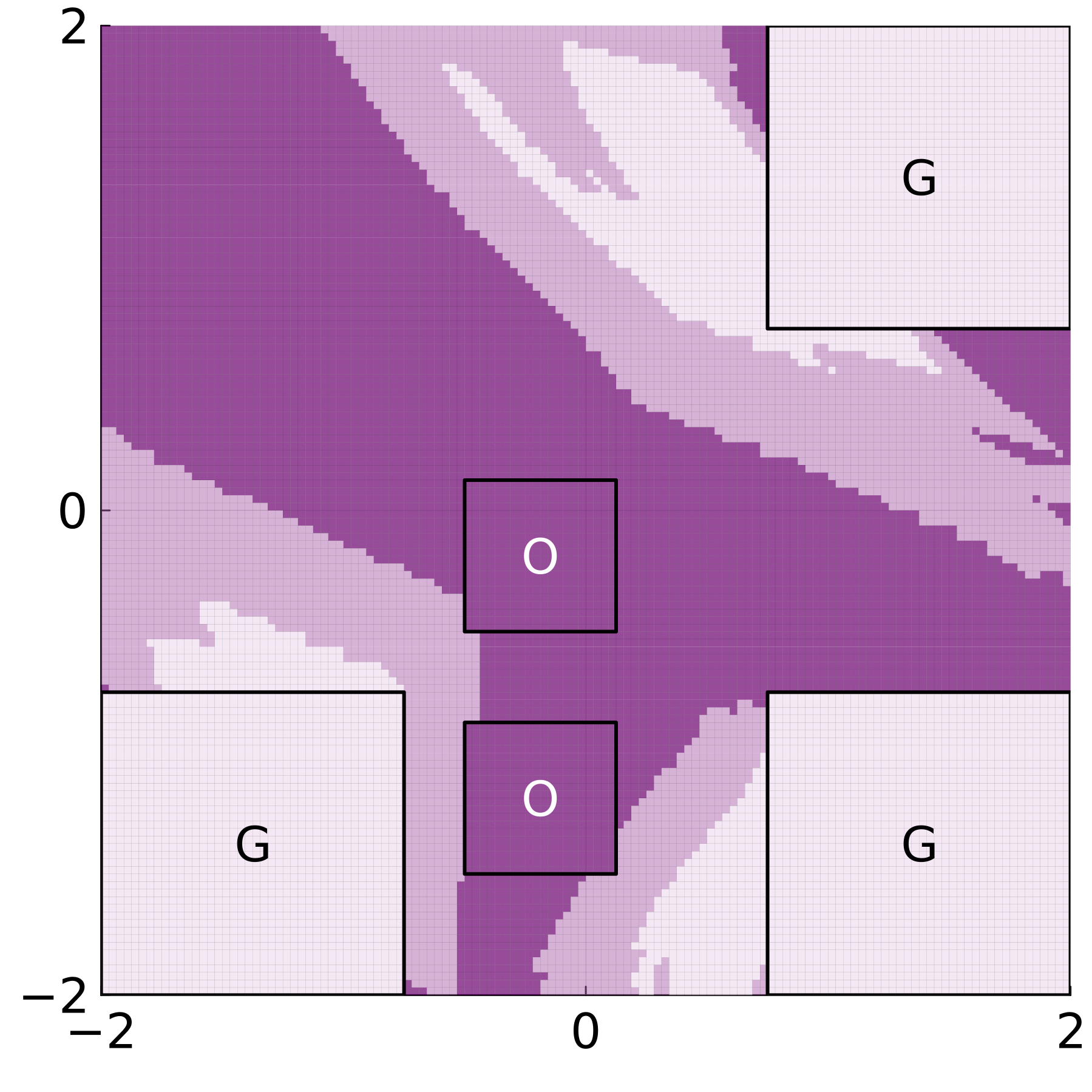}
    \vspace{\vertspace}
    \caption{$k=2$}
    \label{fig:nl5-2-step}
    \end{subfigure}
    \begin{subfigure}{\figwidth}
    \includegraphics[width=\linewidth]{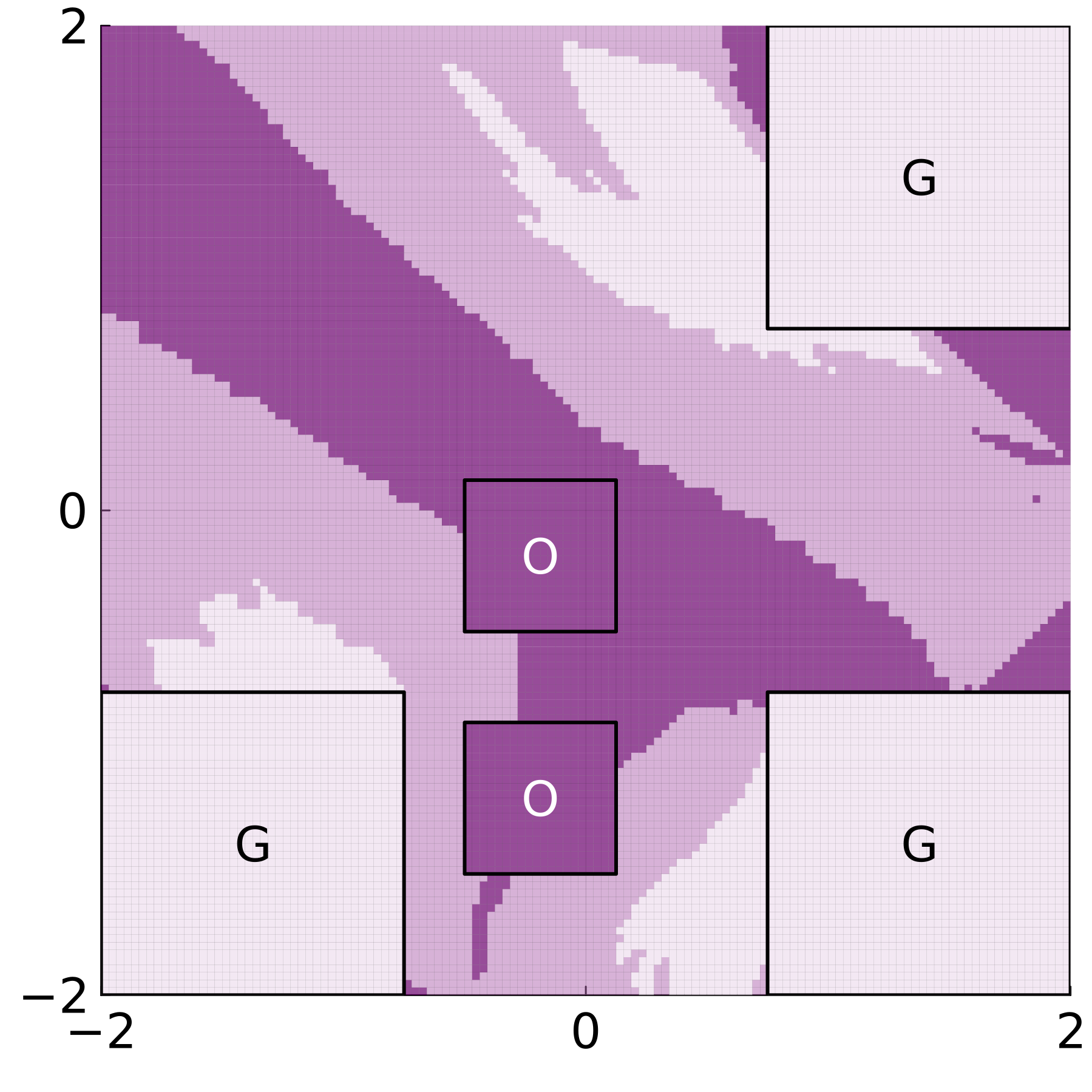}
    \vspace{\vertspace}
    \caption{$k=3$}
    \label{fig:nl5-3-step}
    \end{subfigure}
    \begin{subfigure}{\figwidth}
    \includegraphics[width=\linewidth,trim={1.5cm, 1.5cm, 0.75cm, 1.5cm}, clip]{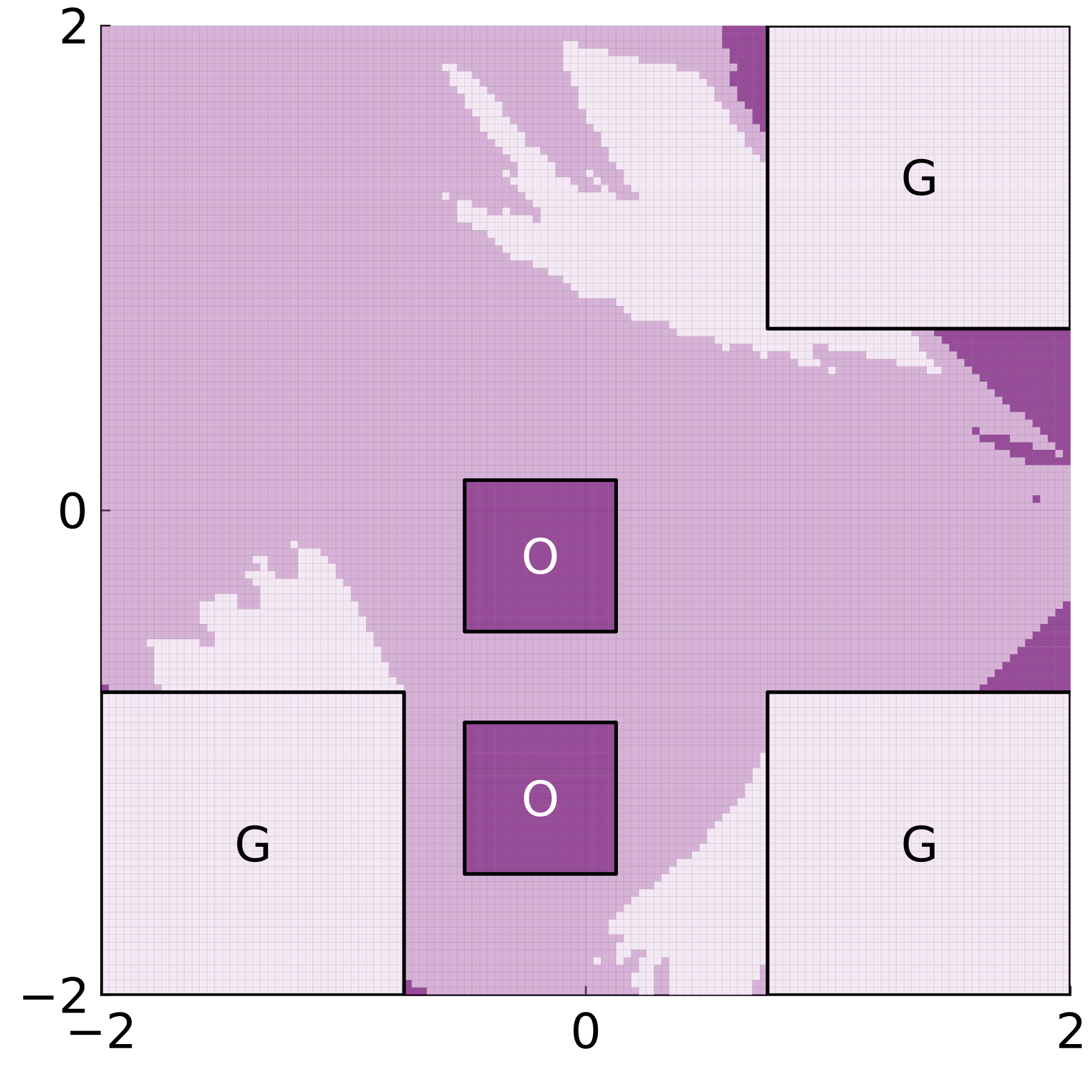}
    \vspace{\vertspace}
    \caption{$k=\infty$}
    \label{fig:nl-inf-step}
    \end{subfigure}
    \caption{Verification results for the nonlinear system subject to $\Spec=\Pr_{\geq 0.95}[O~\BoundedUntilOp G]$.
    }
    \label{fig:nl5}
\end{figure}

A 2000-point dataset is used for GP regression with the learned vector fields shown in Fig.~\ref{fig:nl5-learned}.
The specification $\Spec=\Pr_{\geq 0.95}[\neg O~\BoundedUntilOp G]$ commands to avoid $O$ and remain within $\FullSet$ until $G$ is reached within $k$ steps. 
With $\Delta=0.03125$, the uniform discretization yields 16,384 discrete states which culminated into an end-to-end verification runtime of 38 minutes.
The majority of this time involved computing the images and error bounds for each discrete state, which is the only parallelized part of the framework.

For $k=1$ shown in ~\ref{fig:nl5-1-step}, $\Qyes$ consists of the $G$ regions and nearby states that transition to the $G$ regions with high probability.
The size of the $\Qposs$ buffer between $\Qyes$ and $\Qno$ provides a qualitative measure of the uncertainty embedded in the abstraction. 
As $k$ increases, $\Qyes$ modestly grows while $\Qno$ shrinks.
The $k=\infty$ result shown in Fig.~\ref{fig:nl-inf-step} indicates where guarantees of satisfying or violating $\phi$ can be made over an unbounded horizon. 
There are many more states in $\Qyes$ as compared to even the $k=3$ result.
The increase of $\Qposs$ from $k=1$ to $k=\infty$ is due to the propagating learning and discretization uncertainty over longer horizons.
The discretization resolution plays an important role in the size of $\Qposs$, as the uncertain regions are easily induced by discrete states with possible self-transitions, which is mitigated with a finer discretization.

\subsection{3D Closed-Loop Dubin's Car}

\newcommand{\ImpSpec}{\Pr_{\geq 0.95}\big[\big(\neg O \wedge \big(GF \implies \Pr_{\geq 0.95}[\EventuallyOp~D2]\big)\big) \, \UntilOp ~D1\big]}

\begin{figure}[h!]
    \centering
    \newcommand\vertspace{-7mm}
    \begin{subfigure}{.7\linewidth}
    \centering
    \includegraphics[width=0.5\linewidth]{images/legend.png}
    \end{subfigure}
    \newcommand\figwidth{0.9\linewidth}
    \begin{subfigure}{\figwidth}
    \includegraphics[width=\linewidth,trim={0 2cm 0 2cm},clip]{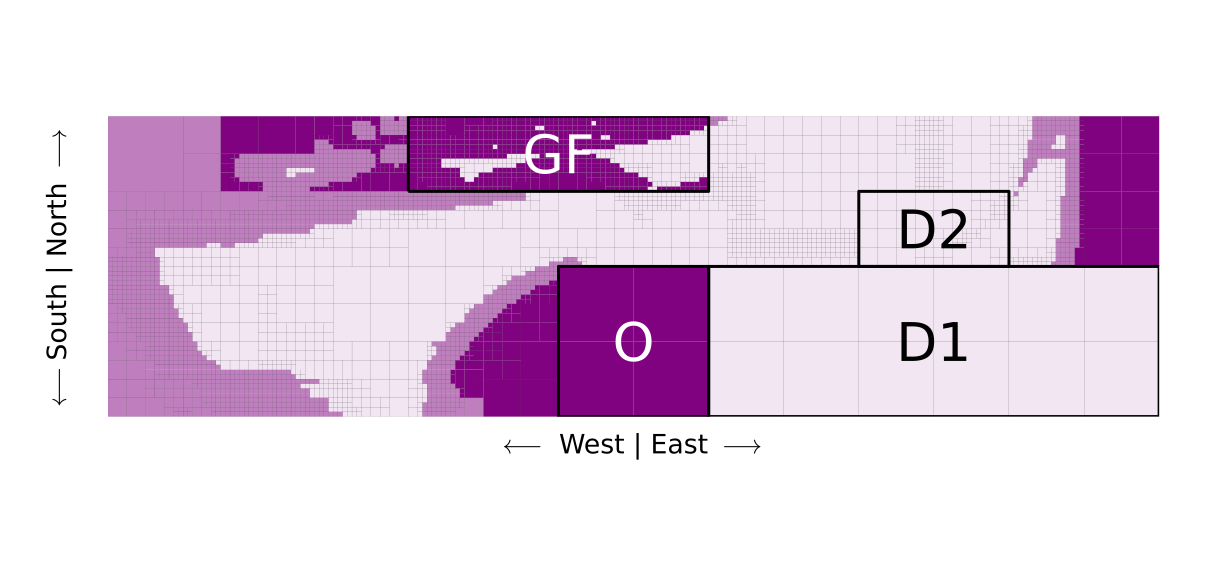}
    \vspace{\vertspace}
    \caption{$\theta_0=\pi/2$ $\rightarrow$}
    \label{fig:dubins1}
    \end{subfigure}
    \begin{subfigure}{\figwidth}
    \includegraphics[width=\linewidth,trim={0 2cm 0 2cm},clip]{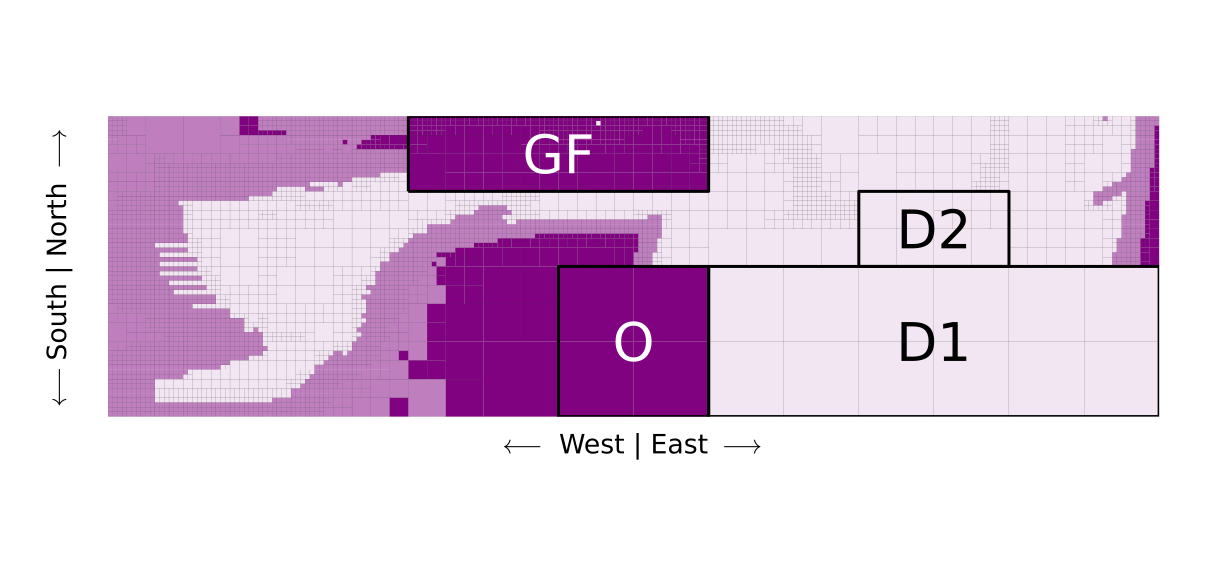}
    \vspace{\vertspace}
    \caption{$\theta_0=9\pi/10$ \rotatebox[origin=c]{-72}{$\rightarrow$}}
    \label{fig:dubins2}
    \end{subfigure}
     \begin{subfigure}{\figwidth}
    \includegraphics[width=\linewidth,trim={0 2cm 0 2cm},clip]{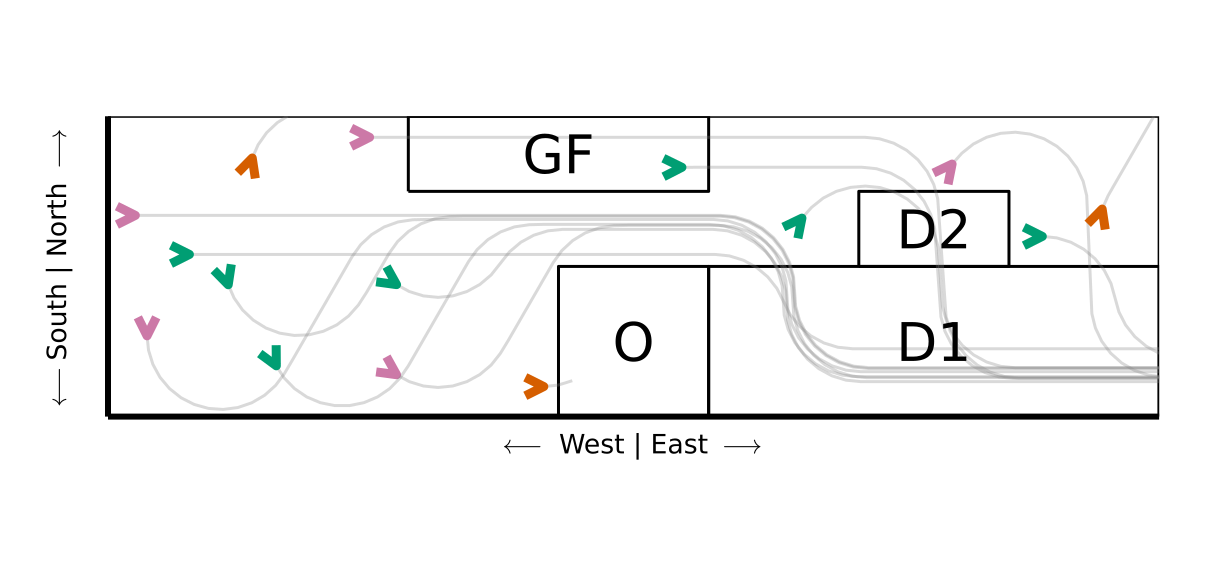}
    \vspace{\vertspace}
    \caption{Trajectories with initial heading colored to indicate prior satisfaction (\textcolor[HTML]{009E73}{$\pmb{\bm{>}}$}), violation (\textcolor[HTML]{D55E00}{$\pmb{\bm{>}}$}), or indeterminant (\textcolor[HTML]{CC79A7}{$\pmb{\bm{>}}$}) guarantees according to the framework.}
    \label{fig:dubins-sim}
    \end{subfigure}
    \caption{Verification results for the Dubin's car system at different initial headings subject to $\ImpSpec$.
    }
    \label{fig:dubins}
\end{figure}
\color{black}

Finally, we consider a discrete-time form of the Dubin's car model  
\begin{equation} 
\begin{bmatrix} x_{k+1} \\
y_{k+1}\\
\theta_{k+1} \\
\end{bmatrix} = 
\begin{bmatrix} x_k\\
y_k\\
\theta_{k} \\
\end{bmatrix} 
+ \Delta T \begin{bmatrix} \sinc(\tilde{u})v \cos(\theta + \tilde{u})\\ 
\sinc(\tilde{u})v \sin(\theta + \tilde{u})\\ 
u \end{bmatrix} \end{equation} 
where $x \in [-1,3]$ and $y \in [0,1]$ are the position, $\theta \in [-0.6, 3.8]$ is the heading angle,
$u$ controls the change in heading angle $\theta$, $\tilde{u}=0.5u\Delta T$, $v$ is the constant speed, and $\Delta T$ is the time discretization parameter.

The system models an aircraft that evolves in the environment shown in Fig.~\ref{fig:dubins-sim} that is equipped with a feedback controller $\pi$ designed to satisfy the following specification:
``avoid obstacle $O$ and reach a goal area $D1$; if the aircraft enters the geofenced area $GF$, it has to enter area $D2$ before reaching $D1$.''
The PCTL formula for verification is
\begin{equation}
    \Spec=\ImpSpec.
\end{equation}
To verify the (unknown) system,
a dataset consisting of 100,000 datapoints was generated. 
To increase the speed of the framework with this large dataset, GP regression was performed \emph{locally}, which has been shown to be improve efficiency in higher dimensions~\cite{jackson2021synergistic}. 
For each discrete region, the nearest 1,000 datapoints according to a suitable $SE(2)$ metric are used for local dynamics modelling.
In the same vein, local RKHS constants are computed for each region.
An automatic refinement procedure was employed to repeatedly refine the subset of $\Qposs$ that may reach $\Qyes$, with 12 refinement steps in total taking 35 hours.
With the initial discretization consisting of 280 states, 21.4\% and 17.5\% of all states by volume belong to $\Qyes$ and $\Qno$ respectively. 
After the 12 refinement steps, the discretization consists of 309,092 states with 49.3\% and 26.4\% belonging to $\Qyes$ and $\Qno$ respectively.

Fig.~\ref{fig:dubins} presents the verification results in the $x$ and $y$ (north and east) dimensions for two initial headings $\theta_0$, although results were generated for all headings within $[-0.6,~3.8]$ radians.
Fig.~\ref{fig:dubins-sim} shows trajectories of the system with controller $\pi$ with indicated initial headings that are colored according to their performance guarantee. 
The classification of each initial state as satisfying or violating the specification intuitively follows a pattern depending on the initial heading. 
When this heading is eastward, as seen in Fig.~\ref{fig:dubins1}, the controller can steer the system around the obstacle $O$ if there is enough clearance.
However, when the initial region is south-eastward, shown in Fig.~\ref{fig:dubins2}, the size of $\Qno$ around the obstacle grows as there is not enough time to avoid the obstacle.

While the controller was designed with the implication $GF \implies \Pr_{\geq 0.95}\big[\mathcal{F}~D2\big]$ in mind, the verification procedure identifies shortcomings in the controller performance. 
For example, when the initial heading is south-eastward (Fig.~\ref{fig:dubins2}) or the system is initialized in $GF$, in most regions there is no guarantee that the system reaches $D2$ with at least 95\% probability. 
Of the two trajectories that pass through $GF$ with $\theta_0=\pi/2$ in Fig.~\ref{fig:dubins-sim}, one is classified as satisfying $\Spec$ and the other classified as indeterminant.
This is due to the eastern regions of $GF$ in Fig.~\ref{fig:dubins1} having a high probability of reaching $D2$, which satisfies the implication.
The western portion of $GF$, in addition to the majority of $GF$ in the other case, does not have a minimum probability that satisfies this relation.
However, this does not preclude the system from actually achieving the task as seen in Fig.~\ref{fig:dubins-sim}, since the maximum probability of reaching $D2$ is above zero but not guaranteed to be at least 95\%. 
In fact, further refinement of the abstraction to reduce the transition uncertainty may result in greater portions of $GF$ satisfying the requirement of reaching $D2$.

\section{Conclusion}
\label{sec:conclusion}

We present a verification framework for unknown dynamic systems with measurement noise subject to PCTL specifications. 
Through learning with GP regression and sound \imdp 
abstraction, we can achieve guarantees on satisfying (or violating) complex properties when the system dynamics are \emph{a priori} unknown.
The maturation of this framework will address the challenges of scaling to higher dimensions, which include the reliance on larger datasets and more discrete states in the abstraction, which can lead to the state-explosion dilemma.

Extensions to the framework could allow for synthesis with measurement noise, general and nonlinear measurement models, and the scaling issues of GP regression and the bounding of state images and posterior covariance suprema.
While effective for GP regression, our interval-based abstraction method may be applied to other types of learning-based approaches provided sound probabilistic errors are available.

\section*{References}
\bibliographystyle{IEEEtran}        
\bibliography{refs,lahijanian}

\appendix

\section*{Proof of Proposition~\ref{prop:rkhsbound}}
\label{app:proof_RKHSnorm}

\begin{proof}
The proof relies on the following lemma, which gives equivalent conditions involving a function $f$ in an RKHS $\RKHS$, its RKHS norm $\|f\|_{\kernel}$, and the associated kernel $\kernel$.
\begin{lemma}[Theorem 3.11 from~\cite{paulsen2016introduction}]\label{lemma:rkhsnorm}
Let $\RKHS$ be an RKHS on $X$ with reproducing kernel $\kernel$ and let $\mathrm{f}:X\to \reals$ be a function. Then the following are equivalent:
\begin{enumerate}%
    \item \label{lem:sub1} $\mathrm{f}\in\RKHS$;
    \item \label{lem:sub2} there exists a constant $c\geq 0$ such that, for every finite subset $F=\{x_1,\dots x_n\}\subseteq X$, there exists a function $h\in\RKHS$ with $\|h\|_{\kernel}\leq c$ and $\mathrm{f}(x_i) = h(x_i)$  for all $i=1,\dots n$;
    \item \label{lem:sub3} there exists a constant $c\geq 0$ such that the function $c^2\kernel(x,x') - \mathrm{f}(x)\mathrm{f}(x')$ is a kernel function.
\end{enumerate}
Moreover, if $\mathrm{f}\in\RKHS$ then $\|\mathrm{f}\|_\kernel$ is the least $c$ that satisfies the inequalities in \ref{lem:sub2} and \ref{lem:sub3}.
\end{lemma}
Let $\mathrm{f}:X\to\reals$ reside in the RKHS $\RKHS$ defined by $\kernel$. 
Note that any $c$ that satisfies \ref{lem:sub3} in Lemma~\ref{lemma:rkhsnorm} also satisfies \ref{lem:sub2}. due to the assumption that $\mathrm{f}\in\RKHS$ and choosing $h=\mathrm{f}$.

A kernel function $\kernel$ is positive semidefinite if and only if, for all finite subsets $\{x_1,\dots x_n\}\subset X$ and real vectors $\{b_1,\dots b_n\}\in\reals^n$,
$
\sum_{i}\sum_{j} \kernel(x_i,x_j)b_ib_j\geq 0
$
which the given kernel satisfies. 
By Lemma~\ref{lemma:rkhsnorm}, there exists a constant $c\geq 0$ such that
\begin{align*}
    \sum_{i} & \sum_{j} (c^2\kernel(x_i,x_j) - \mathrm{f}(x_i)\mathrm{f}(x_j))b_ib_j \\
    &=\sum_{i}\sum_{j} \kernel(x_i,x_j)b_ib_j - c^{-2}\mathrm{f}(x_i)\mathrm{f}(x_j)b_ib_j\\
    &=\sum_{i}\sum_{j} \kernel(x_i,x_j)b_ib_j - \sum_{i}\sum_{j}c^{-2}\mathrm{f}(x_i)\mathrm{f}(x_j)b_ib_j 
    \geq 0
\end{align*}
which leads to
\begin{align*}
&\sum_{i}\sum_{j} \kernel(x_i,x_j)b_ib_j \geq
\sum_{i}\sum_{j}c^{-2}\mathrm{f}(x_i)\mathrm{f}(x_j)b_ib_j 
\end{align*}
Assume $\kernel(x,x')\geq 0$. 
Then choose $c\geq 0$ such that, $\forall i,j$, 
\begin{align*}\label{eq:c}
    \kernel(x_i,x_j) &\geq |c^{-2}\mathrm{f}(x_i)\mathrm{f}(x_j)|\\
    c^2 &\geq \frac{|\mathrm{f}(x_i)\mathrm{f}(x_j)|}{\kernel(x_i,x_j)}
\end{align*}
which obligates the choice 
\begin{equation}
    \label{eq:cmax}
    c^2 = \sup_{x,x'\in X}\frac{|\mathrm{f}(x)\mathrm{f}(x')|}{\kernel(x,x')}
    \leq \frac{\sup_{x\in X} f(x)^2}{\inf_{x,x'\in X}\kernel(x,x')}. 
\end{equation}
By Lemma~\ref{lemma:rkhsnorm}, the right-hand side is an upper bound of $\|\mathrm{f}\|^2_\kernel$.
\end{proof}

\section*{Proof of Proposition~\ref{prop:infobound}}
\label{app:proof_info_gain}
\begin{proof}
The proof relies on Hadamard's determinant bounding inequality and $s=\sup_{x\in X}\kernel(x,x)$.
Let $K$ denote the kernel matrix from $d$ input points and kernel $\kappa$. 
Then, the maximum information gain is bounded by
\begin{align*}
    \InfoGainBound \coloneq \max_{K} \log |I + \sigma^{-2}K| &\leq \max_K \log \prod_{i=1}^{\datasetsize} |1 + \sigma^{-2}K_{i,i}| \\
       \leq \log \prod_{i=1}^{\datasetsize} |1 + \sigma^{-2}s|
       &= \sum_{i=1}^{\datasetsize} \log (1+\sigma^{-2}s)\\
      & = \datasetsize\log(1+\sigma^{-2}s)
\end{align*}
where $|\cdot|$ is the matrix determinant function.
\end{proof}

\section*{Proof of Proposition~\ref{prop:distbound}}
\label{app:proof_optimal_error}

\begin{proof}
Let $\Pr(\distBound_i)=(1-\delta_i(\distBound_i))$ denote the CDF of the probabilistic regression error. 
In order to maximize the upper bound in Theorem~\ref{Th:TransitionBounds}, the indicator function should return 1 and the product with $\Pr(\distBound_i)$ should be maximizing. The indicator function returns 1 if $\qimdppost\cap\underline{\qimdpprime}=\qimdppost$. The most $\qimdpprime$ can shrink is $$\inf_{a\in \partial \qimdppost,b\in \partial \qimdpprime} |a-b|.$$ Since $\Pr(\distBound_i)$ is non-decreasing, the choice of $\distBound$ is maximizing. The proof for the minimizing case is similar. 
\end{proof}
\section*{Proof of Lemma~\ref{lemma:ValueIterOriginalSys}}
\label{app:proof-valueIterOriginalSys}

\begin{proof}
Assume, for simplicity and w.l.o.g., that $\strX$ is stationary. Then $V_k^{\strX}:\reals^n \to \reals $ is defined recursively as
        \begin{align*}
            V_k^{\strX}(x)=
            \begin{cases}
                1 & \text{if } x \in \FullSet^1\\
                0 & \text{if } x \in \FullSet^0\\
                0 & \text{if } x \notin (\FullSet^0 \cup \FullSet^1) \wedge k = 0\\
                \mathrlap{\mathbf{E}_{v\sim \noisedistribution, \, \bar{x}\sim T(\cdot \mid x,\strX(x+v) )} [ V_{k-1}(\bar{x})]} & \hspace{45mm}\text{otherwise.}
            \end{cases}
        \end{align*}    
To prove Lemma~\ref{lemma:ValueIterOriginalSys},  we need to show that $$ V_k^{\strX}(x) = \Pr(\omega^k_{\x}\Satisfies\PathFormula\mid \omega_{\x}(0) = x,\strX),$$
that is the probability that a path of length $k$ of Process \eqref{eq:system} initialized at $x$ is equal to $V_k^{\strX}(x).$

The proof is by induction over the length of the path. The base case is $$ \Pr(\omega^0_{\x}\Satisfies\PathFormula\mid \omega^0_{\x}(0) = x,\strX)=\mathbf{1}_{\FullSet^1}(x)=V_{0}^{\strX}(x).$$
Then, to conclude the proof we need to show that under the assumption that for all $x\in \mathbb{R}^n$
$$ \Pr(\omega^{k-1}_{\x}\Satisfies\PathFormula\mid \omega_{\x}(0) =x)=V_{k-1}^{\strX}(x), $$
it holds that for all $x\in \mathbb{R}^n$
$$\Pr(\omega^k_{\x}\Satisfies\PathFormula\mid \omega_{\x}(0) = x,\strX)=V_{k}^{\strX}(x) .$$
Call $\bar{\FullSet}=\FullSet \setminus (\FullSet^0 \cup \FullSet^1),$ that is the complement of set $\FullSet^0 \cup \FullSet^1$, and define notation 
$$\omega^k_{\x}([i,k-1])\in \FullSet:=\forall j\in [i,k-1],\,\omega^k_{\x}(j)\in \FullSet.  $$
Assume w.l.o.g. that $x\in \bar{\FullSet}$.
The final part of the proof is:
\begin{align*}
   & \Pr(\omega^k_{\x}\Satisfies\PathFormula\mid \omega_{\x}(0) = x,\strX)\\
    &= \sum_{i=1}^k \Pr(\omega^k_{\x}(i)\in \FullSet^1 \wedge  \omega^k_{\x}([1,i-1])\in \bar{\FullSet} \mid \omega^k_{\x}(0)=x,\strX) \\
    &= \Pr(\omega^k_{\x}(1)\in \FullSet^1\mid \omega^k_{\x}(0)=x,\strX)\,+\\
    & \hspace{5mm} \sum_{i=2}^k \Pr(\omega^k_{\x}(i)\in \FullSet^1 \wedge \omega^k_{\x}([1,i-1])\in \bar{\FullSet}\mid \omega^k_{\x}(0)=x,\strX) \\
    &= \int \hspace{-2mm} \int_{\FullSet^1} t(\bar{x}\mid x,\strX(x+v))\noisedistribution(v)d\bar{x}dv \,+ \\
    & \quad \sum_{i=2}^k \int \hspace{-2mm} \int_{\bar{\FullSet}} \Pr(\omega^k_{\x}(i)\in \FullSet^1 \wedge \omega^k_{\x}([2,i-1])\in \bar{\FullSet}\mid \\
    &\hspace{2cm} \omega^k_{\x}(1)=\bar{x},\strX) \, t(\bar{x}\mid x,\strX(x+v))\noisedistribution(v) d\bar{x}dv\\
    &= \int \hspace{-2mm} \int_{\FullSet^1} t(\bar{x}\mid x,\strX(x+v))\noisedistribution(v)d\bar{x}dv \, + \\
    & \quad \sum_{i=1}^{k-1} \int \hspace{-2mm} \int_{\bar{\FullSet}} \Pr\big(\omega^{k-1}_{\x}(i)\in \FullSet^1 \wedge \omega^{k-1}_{\x}([1,i-1])\in \bar{\FullSet}\mid \\
    &\hspace{1.6cm} \omega^{k-1}_{\x}(0)=\bar{x},\strX\big) \, t(\bar{x}\mid x,\strX(x+v))\noisedistribution(v) d\bar{x}dv\\
    &= \int \hspace{-2mm} \int \big( \mathbf{1}_{\FullSet^1 }(\bar{x}) + \mathbf{1}_{\bar{\FullSet}}(\bar{x}) \, \Pr(\omega^{k-1}_{\x}\Satisfies\PathFormula \mid \omega_{\x}(0) = \bar{x},\strX) \big)\\
    & \hspace{4.3cm} t(\bar{x}\mid x,\strX(x+v))\noisedistribution(v) d\bar{x}dv\\
    &= \int \hspace{-2mm} \int  \Pr(\omega^{k-1}_{\x}\Satisfies\PathFormula\mid \omega_{\x}(0) = \bar{x},\strX)  \\
    & \hspace{4.3cm} t(\bar{x}\mid x,\strX(x+v))\noisedistribution(v) d\bar{x}dv\\
    &= V_k^{\strX}(x),
\end{align*}
where in the third equality we marginalized over the event $\omega^k_{\x}(1)=\bar{x} $ and use the definition of conditional probability, and in the fourth equality we use that for $i<k$ and $x\in \bar{\FullSet}$
\begin{align*}
    &\Pr(\omega^{k-1}_{\x}(i)\in \FullSet^1 \wedge \omega^{k-1}_{\x}([1,i-1])\in \bar{\FullSet}\mid \omega^{k-1}_{\x}(0)=\bar{x},\strX) \\
    &=\Pr(\omega^{k}_{\x}(i+1)\in \FullSet^1 \wedge \omega^{k}_{\x}([2,i+1-1])\in \bar{\FullSet}\mid \\
    & \hspace{6.5cm} \omega^{k}_{\x}(1)=\bar{x},\strX),
\end{align*}
which holds because of the Markov property.

\end{proof}

\begin{IEEEbiographynophoto}{John Skovbekk} (Student Member, IEEE) is a PhD candidate in the Ann and H.J. Smead Department of Aerospace Engineering Sciences at the University of Colorado, Boulder. He received Bachelor of Aerospace Engineering and Mechanics (2015) and Master of Science (2018) degrees from the University of Minnesota, Twin Cities. His research focuses on the intersection of formal control methods with data-driven modeling to enhance the safety and efficacy of robotics.
\end{IEEEbiographynophoto}

\begin{IEEEbiographynophoto}{Luca Laurenti} (Member, IEEE) is a tenure track assistant professor at the Delft Center for Systems and Control at TU Delft and co-director of the  HERALD Delft AI Lab . He received his PhD from the Department of Computer Science, University of Oxford (UK), where he was a member of the Trinity College. Luca has a background in stochastic systems, control theory, formal methods, and artificial intelligence. His research work focuses on developing data-driven systems provably robust to interactions with a dynamic and uncertain world.
\end{IEEEbiographynophoto}

\begin{IEEEbiographynophoto}{Prof. Eric Frew} (Member, IEEE) is a professor in the Ann and H.J. Smead Aerospace Engineering Sciences Department at the University of Colorado Boulder (CU). He received his B.S. in mechanical engineering from Cornell University in 1995 and his M.S and Ph.D. in aeronautics and astronautics from Stanford University in 1996 and 2003, respectively. Dr. Frew has been designing and deploying autonomous robotic systems for over twenty five years. His research efforts focus on autonomous flight of heterogeneous uncrewed aircraft systems; distributed information-gathering by mobile robots; miniature self-deploying systems; and guidance and control of uncrewed aircraft in complex atmospheric phenomena. Dr. Frew has successfully deployed autonomous robots in a wide variety of scenarios. He was co-leader of the team that performed the first-ever sampling of a severe supercell thunderstorm by an uncrewed aircraft and that recently deployed two autonomous coordinating aircraft in a supercell. He was also co-leader of the CU MARBLE team that placed third in the DARPA Subterranean Challenge. He is currently the Center Director for the National Science Foundation Industry University Cooperative Research Center (IUCRC) for Autonomous Air Mobility and Sensing (CAAMS).
\end{IEEEbiographynophoto}

\begin{IEEEbiographynophoto}{Morteza Lahijanian} (Member, IEEE) is an assistant professor in the Aerospace Engineering Sciences department, an affiliated faculty at the Computer Science department, and the director of the Assured, Reliable, and Interactive Autonomous (ARIA) Systems group at the University of Colorado Boulder. He received a B.S. in Bioengineering at the University of California, Berkeley and a PhD in Mechanical Engineering at Boston University. He served as a postdoctoral scholar in Computer Science at Rice University. Prior to joining CU Boulder, he was a research scientist in the department of Computer Science at the University of Oxford. His awards include Ella Mae Lawrence R. Quarles Physical Science Achievement Award, Jack White Engineering Physics Award, and NSF GK-12 Fellowship. Dr. Lahijanian's research interests span the areas of control theory, stochastic hybrid systems, formal methods, machine learning, and game theory with applications in robotics, particularly, motion planning, strategy synthesis, model checking, and human-robot interaction. 
\end{IEEEbiographynophoto}

\end{document}